\documentclass[11pt]{article}%
\usepackage{amsfonts}
\usepackage{amsmath}
\usepackage{fullpage}
\usepackage{amssymb}
\usepackage{graphicx}
\usepackage[caption=false]{subfig}%
\usepackage[colorlinks=true,linkcolor=blue,citecolor=red,plainpages=false,pdfpagelabels]%
{hyperref}%
\providecommand{\U}[1]{\protect\rule{.1in}{.1in}}
\newtheorem{theorem}{Theorem}

\newtheorem{lemma}[theorem]{Lemma}

\newtheorem{proposition}[theorem]{Proposition}
\newtheorem{remark}[theorem]{Remark}

\newenvironment{proof}[1][Proof]{\noindent\textbf{#1.} }{\ \rule{0.5em}{0.5em}}
\numberwithin{equation}{section}
\begin{document}

\title{\textbf{Position-based coding and convex splitting for private communication
over quantum channels}}
\author{Mark M. Wilde\thanks{Hearne Institute for Theoretical Physics, Department of
Physics and Astronomy, Center for Computation and Technology, Louisiana State
University, Baton Rouge, Louisiana 70803, USA}}
\date{March 5, 2017}
\maketitle

\begin{abstract}
The classical-input quantum-output (cq)\ wiretap channel is a communication
model involving a classical sender $X$, a legitimate quantum receiver $B$, and
a quantum eavesdropper $E$. The goal of a private communication protocol that
uses such a channel is for the sender $X$ to transmit a message in such a way
that the legitimate receiver $B$\ can decode it reliably, while the
eavesdropper $E$\ learns essentially nothing about which message was
transmitted. The $\varepsilon$-one-shot private capacity of a cq wiretap
channel is equal to the maximum number of bits that can be transmitted over
the channel, such that the privacy error is no larger than $\varepsilon
\in(0,1)$. The present paper provides a lower bound on the $\varepsilon
$-one-shot private classical capacity, by exploiting the recently developed
techniques of Anshu, Devabathini, Jain, and Warsi, called position-based
coding and convex splitting. The lower bound is equal to a difference of the
hypothesis testing mutual information between $X$ and $B$ and the
\textquotedblleft alternate\textquotedblright\ smooth max-information between
$X$ and $E$. The one-shot lower bound then leads to a non-trivial lower bound
on the second-order coding rate for private classical communication over a
memoryless cq wiretap channel.

\end{abstract}

\section{Introduction}

Among the many results of information theory, the ability to use the noise in
a wiretap channel for the purpose of private communication stands out as one
of the great conceptual insights \cite{W75}. A classical wiretap channel is
modeled as a conditional probability distribution $p_{Y,Z|X}$, in which the
sender Alice has access to the input $X$ of the channel, the legitimate
receiver Bob has access to the output $Y$, and the eavesdropper Eve has access
to the output $Z$. The goal of private communication is for Alice and Bob to
use the wiretap channel in such a way that Alice communicates a message
reliably to Bob, while at the same time, Eve should not be able to determine
which message was transmitted. The author of \cite{W75} proved that the mutual
information difference%
\begin{equation}
\max_{p_{X}}\left[  I(X;Y)-I(X;Z)\right]  \label{eq:wiretap-MI-diff}%
\end{equation}
is an achievable rate for private communication over the wiretap channel, when
Alice and Bob are allowed to use it many independent times. Since then, the
interest in the wiretap channel has not waned, and there have been many
increasingly refined statements about achievable rates for private
communication over wiretap channels \cite{CK78,H06,T12,H13,YAG13,YSP16,TB16}.

Many years after the contribution of \cite{W75}, the protocol of quantum key
distribution was developed as a proposal for private communication over a
quantum channel \cite{bb84}. Quantum information theory started becoming a
field in its own right, during which many researchers revisited several of the
known results of Shannon's information theory under a quantum lens. This was
not merely an academic exercise: doing so revealed that remarkable
improvements in communication rates could be attained for physical channels of
practical interest if quantum-mechanical strategies are exploited
\cite{GGLMSY04}.

One important setting which was revisited is the wiretap channel, and in the
quantum case, the simplest extension of the classical model is given by the
classical-input quantum-output wiretap channel (abbreviated as~\textit{cq
wiretap channel}) \cite{ieee2005dev,1050633}. It is described as the following
map:%
\begin{equation}
x\rightarrow\rho_{BE}^{x}, \label{eq:cq-wiretap}%
\end{equation}
where $x$ is a classical symbol that Alice can input to the channel and
$\rho_{BE}^{x}$ is the joint output quantum state of Bob and Eve's system,
represented as a density operator acting on the tensor-product Hilbert space
of Bob and Eve's quantum systems. The goal of private communication over the
cq wiretap channel is similar to that for the classical wiretap channel.
However, in this case, Bob is allowed to perform a collective quantum
measurement over all of his output quantum systems in order to determine
Alice's message, while at the same time, we would like for it be difficult for
Eve to figure out anything about the transmitted message, even if she has
access to a quantum computer memory that can store all of the quantum systems
that she receives from the channel output. The authors of
\cite{ieee2005dev,1050633} independently proved that a quantum generalization
of the formula in \eqref{eq:wiretap-MI-diff} is an achievable rate for private
communication over a cq quantum wiretap channel, if Alice and Bob are allowed
to use it many independent times. Namely, they proved that the following
Holevo information difference is an achievable rate:%
\begin{equation}
\max_{p_{X}}\left[  I(X;B)-I(X;E)\right]  , \label{eq:cq-wiretap-Holevo-infos}%
\end{equation}
where the information quantities in the above formula are the Holevo
information to Bob and Eve, respectively, and will be formally defined later
in the present paper.

Since the developments of \cite{ieee2005dev,1050633}, there has been an
increasing interest in the quantum information community to determine refined
characterizations of communication tasks
\cite{TH12,li12,TT13,DTW14,DHO16,DL15,BDL15,TBR15,WTB16}, strongly motivated
by the fact that it is experimentally difficult to control a large number of
quantum systems, and in practice, one has access only to a finite number of
quantum systems anyway. One such scenario of interest, as discussed above, is
the quantum wiretap channel. Hitherto, the only work offering achievable
one-shot rates for private communication over cq wiretap channels is
\cite{RR11}. However, that work did not consider bounding the second-order
coding rate for private communication over the cq wiretap channel.

The main contribution of the present paper is a lower bound on the one-shot
private capacity of a cq wiretap channel. Namely, I prove that%
\begin{equation}
\log_{2}M_{\operatorname{priv}}^{\ast}(\varepsilon_{1}+\sqrt{\varepsilon_{2}%
})\geq I_{H}^{\varepsilon_{1}-\eta_{1}}(X;B)-\widetilde{I}_{\max}%
^{\sqrt{\varepsilon_{2}}-\eta_{2}}(E;X)-\log_{2}(4\varepsilon_{1}/\eta_{1}%
^{2})-2\log_{2}( 1/\eta_{2}) . \label{eq:1-shot-private-bnd}%
\end{equation}
In the above, $\log_{2}M_{\operatorname{priv}}^{\ast}(\varepsilon_{1}%
+\sqrt{\varepsilon_{2}})$ represents the maximum number of bits that can be
sent from Alice to Bob, using a cq wiretap channel once, such that the privacy
error (to be defined formally later)\ does not exceed $\varepsilon_{1}%
+\sqrt{\varepsilon_{2}}\in(0,1)$, with $\varepsilon_{1},\varepsilon_{2}%
\in(0,1)$. The quantities on the right-hand side of the above inequality are
particular one-shot generalizations of the Holevo information to Bob and Eve,
which will be defined later. It is worthwhile to note that the one-shot
information quantities in \eqref{eq:1-shot-private-bnd} can be computed using
semi-definite programming, and the computational runtime is polynomial in the
dimension of the channel. Thus, for channels of reasonable dimension, the
quantities can be efficiently estimated numerically. The constants $\eta_{1}$
and $\eta_{2}$ are chosen so that $\eta_{1}\in(0,\varepsilon_{1})$ and
$\eta_{2}\in(0,\sqrt{\varepsilon_{2}})$. By substituting an independent and
identically distributed (i.i.d.) cq wiretap channel into the right-hand side
of the above inequality, using second-order expansions for the one-shot Holevo
informations \cite{TH12,li12}, and picking $\eta_{1},\eta_{2}=1/\sqrt{n}$, we
find the following lower bound on the second-order coding rate for private
classical communication:%
\begin{multline}
\log_{2}M_{\operatorname{priv}}^{\ast}(n,\varepsilon_{1}+\sqrt{\varepsilon
_{2}})\geq n\left[  I(X;B)-I(X;E)\right] \\
+\sqrt{nV(X;B)}\Phi^{-1}(\varepsilon_{1})+\sqrt{nV(X;E)}\Phi^{-1}%
(\varepsilon_{2})+O(\log n).
\end{multline}
In the above, $\log_{2}M_{\operatorname{priv}}^{\ast}(n,\varepsilon_{1}%
+\sqrt{\varepsilon_{2}})$ represents the maximum number of bits that can be
sent from Alice to Bob, using a cq wiretap channel $n$ times, such that the
privacy error does not exceed $\varepsilon_{1}+\sqrt{\varepsilon_{2}}\in
(0,1)$. The Holevo informations from \eqref{eq:cq-wiretap-Holevo-infos} make
an appearance in the first-order term (proportional to the number $n$ of
channel uses) on the right-hand side above, while the second order term
(proportional to $\sqrt{n}$) consists of the quantum channel dispersion
quantities $V(X;B)$ and $V(X;E)$ \cite{TT13}, which will be defined later.
They additionally feature the inverse $\Phi^{-1}$ of the cumulative Gaussian
distribution function $\Phi$. Thus, the one-shot bound in
\eqref{eq:1-shot-private-bnd} leads to a lower bound on the second-order
coding rate, which is comparable to bounds that have appeared in the classical
information theory literature \cite{T12,YAG13,YSP16,TB16}.

To prove the one-shot bound in \eqref{eq:1-shot-private-bnd}, I use two recent
and remarkable techniques:\ position-based coding \cite{AJW17}\ and convex
splitting \cite{ADJ17}. The main idea of position-based coding \cite{AJW17}%
\ is conceptually simple. To communicate a classical message from Alice to
Bob, we allow them to share a quantum state $\rho_{RA}^{\otimes M}$ before
communication begins, where $M$ is the number of messages, Bob possesses the
$R$ systems, and Alice the $A$ systems. If Alice wishes to communicate message
$m$, then she sends the $m$th $A$ system through the channel. The reduced
state of Bob's systems is then%
\begin{equation}
\rho_{R_{1}}\otimes\cdots\otimes\rho_{R_{m-1}}\otimes\rho_{R_{m}B}\otimes
\rho_{R_{m+1}}\otimes\cdots\otimes\rho_{R_{M}},
\label{eq:position-based-decoding}%
\end{equation}
where $\rho_{R_{m}B}=\mathcal{N}_{A_{m}\rightarrow B}(\rho_{R_{m}A_{m}})$ and
$\mathcal{N}_{A_{m}\rightarrow B}$ is the quantum channel. For all $m^{\prime
}\neq m$, the reduced state for systems $R_{m^{\prime}}$ and $B$ is the
product state $\rho_{R_{m^{\prime}}}\otimes\rho_{B}$. However, the reduced
state of systems $R_{m}B$ is the (generally)\ correlated state $\rho_{R_{m}B}%
$. So if Bob has a binary measurement which can distinguish the joint state
$\rho_{RB}$ from the product state $\rho_{R}\otimes\rho_{B}$ sufficiently
well, he can base a decoding strategy off of this, and the scheme will be
reliable as long as the number of bits $\log_{2}M$ to be communicated is
chosen to be roughly equal to a one-shot mutual information known as
hypothesis testing mutual information (cf., \cite{WR12}). This is exactly what
is used in position-based coding, and the authors of \cite{AJW17} thus forged
a transparent and intuitive link between quantum hypothesis testing and
communication for the case of entanglement-assisted communication.

Convex splitting \cite{ADJ17} is rather intuitive as well and can be thought
of as dual to the coding scenario mentioned above. Suppose instead that Alice
and Bob have a means of generating the state in
\eqref{eq:position-based-decoding}, perhaps by the strategy mentioned above.
But now suppose that Alice chooses the variable $m$ uniformly at random, so
that the state, from the perspective of someone ignorant of the choice of $m$,
is the following mixture:%
\begin{equation}
\frac{1}{M}\sum_{m=1}^{M}\rho_{R_{1}}\otimes\cdots\otimes\rho_{R_{m-1}}%
\otimes\rho_{R_{m}B}\otimes\rho_{R_{m+1}}\otimes\cdots\otimes\rho_{R_{M}}.
\end{equation}
The convex-split lemma guarantees that as long as $\log_{2}M$ is roughly equal
to a one-shot mutual information known as the alternate smooth max-mutual
information, then the state above is nearly indistinguishable from the
product state $\rho_{R}^{\otimes M}\otimes\rho_{B}$.

Both position-based coding and convex splitting have been used recently and
effectively to establish a variety of results in one-shot quantum information
theory \cite{AJW17,ADJ17}. In the present paper, I use the approaches in
conjunction to construct codes for the cq wiretap channel. The main underlying
idea follows the original approach of \cite{W75}, by allowing for a message
variable $m\in\{1,\ldots,M\}$ and a local key variable $k\in\{1,\ldots,K\}$
(local randomness), the latter of which is selected uniformly at random and
used to confuse the eavesdropper Eve. Before communication begins, Alice, Bob,
and Eve are allowed share to $MK$ copies of the common randomness state
$\theta_{X_{A}X_{B}X_{E}}\equiv\sum_{x}p_{X}(x)|xxx\rangle\langle
xxx|_{X_{A}X_{B}X_{E}}$. We can think of the $MK$ copies of $\theta
_{X_{A}X_{B}X_{E}}$ as being partitioned into $M$ blocks, each of which
contain $K$ copies of the state $\theta_{X_{A}X_{B}X_{E}}$. If Alice wishes to
send message $m$, then she picks $k$ uniformly at random and sends the $(m,k)$
$X_{A}$ system through the cq wiretap channel in \eqref{eq:cq-wiretap}. As
long as $\log_{2}MK$ is roughly equal to the hypothesis testing mutual
information $I_{H}^{\varepsilon}(X;B)$, then Bob can use a position-based
decoder to figure out both $m$ and $k$. As long as $\log_{2}K$ is roughly
equal to the alternate smooth max-mutual information $\widetilde{I}_{\max
}^{\varepsilon}(E;X)$, then the convex-split lemma guarantees that the overall
state of Eve's systems, regardless of which message $m$ was chosen, is nearly
indistinguishable from the product state $\rho_{X_{E}}^{\otimes MK}\otimes
\rho_{E}$. Thus, in such a scheme, Bob can figure out $m$ while Eve cannot
figure out anything about $m$. This is the intuition behind the coding scheme
and gives a sense of why $\log_{2}M=\log_{2}MK-\log_{2}K\approx I_{H}%
^{\varepsilon}(X;B)-\widetilde{I}_{\max}^{\varepsilon}(E;X)$ is an achievable
number of bits that can be sent privately from Alice to Bob. The main purpose
of the present paper is to develop the details of this argument and
furthermore show how the scheme can be derandomized, so that the $MK$ copies
of the common randomness state $\theta_{X_{A}X_{B}X_{E}}$ are in fact not necessary.

The rest of the paper proceeds as follows. In Section~\ref{sec:prelim}, I
review some preliminary material, which includes several metrics for quantum
states and pertinent information measures. Section~\ref{sec:public-comm}
develops the position-based coding approach for classical-input quantum-output
communication channels. Position-based coding was developed in \cite{AJW17}%
\ to highlight a different approach to entanglement-assisted communication,
but I show in Section~\ref{sec:public-comm} how the approach can be used for
shared randomness-assisted communication; I also show therein how to
derandomize codes in this case (i.e., the shared randomness is not actually
necessary for classical communication over cq channels).
Section~\ref{sec:priv-comm} represents the main contribution of the present
paper, which is a lower bound on the $\varepsilon$-one-shot private classical
capacity of a cq wiretap channel. The last development in
Section~\ref{sec:priv-comm} is to show how the one-shot lower bound leads to a
lower bound on the second-order coding rate for private classical
communication over a memoryless cq wiretap channel. 
Therein, I also show how these lower bounds simplify for pure-state cq wiretap channels and when using binary phase-shift keying as a coding strategy for private communication over a pure-loss bosonic channel.
Section~\ref{sec:concl}%
\ concludes with a summary and some open questions for future work.

\section{Preliminaries\label{sec:prelim}}

I use notation and concepts that are standard in quantum information theory
and point the reader to \cite{W15book} for background. In the rest of this
section, I review concepts that are less standard and set some notation that
will be used later in the paper.

\bigskip

\textbf{Trace distance, fidelity, and purified distance.} Let $\mathcal{D}%
(\mathcal{H})$ denote the set of density operators acting on a Hilbert space
$\mathcal{H}$, let $\mathcal{D}_{\leq}(\mathcal{H})$ denote the set of
subnormalized density operators (with trace not exceeding one) acting on
$\mathcal{H}$, and let $\mathcal{L}_{+}(\mathcal{H})$ denote the set of
positive semi-definite operators acting on $\mathcal{H}$. The trace distance
between two quantum states $\rho,\sigma\in\mathcal{D}(\mathcal{H})$\ is equal
to $\left\Vert \rho-\sigma\right\Vert _{1}$, where $\left\Vert C\right\Vert
_{1}\equiv\operatorname{Tr}\{\sqrt{C^{\dag}C}\}$ for any operator $C$. It has
a direct operational interpretation in terms of the distinguishability of
these states. That is, if $\rho$ or $\sigma$ are prepared with equal
probability and the task is to distinguish them via some quantum measurement,
then the optimal success probability in doing so is equal to $\left(
1+\left\Vert \rho-\sigma\right\Vert _{1}/2\right)  /2$. The fidelity is
defined as $F(\rho,\sigma)\equiv\left\Vert \sqrt{\rho}\sqrt{\sigma}\right\Vert
_{1}^{2}$ \cite{U76}, and more generally we can use the same formula to define
$F(P,Q)$ if $P,Q\in\mathcal{L}_{+}(\mathcal{H})$. Uhlmann's theorem states
that \cite{U76}%
\begin{equation}
F(\rho_{A},\sigma_{A})=\max_{U}\left\vert \langle\phi^{\sigma}|_{RA}%
U_{R}\otimes I_{A}|\phi^{\rho}\rangle_{RA}\right\vert ^{2},
\label{eq:uhlmann-thm}%
\end{equation}
where $|\phi^{\rho}\rangle_{RA}$ and $|\phi^{\sigma}\rangle_{RA}$ are fixed
purifications of $\rho_{A}$ and $\sigma_{A}$, respectively, and the
optimization is with respect to all unitaries $U_{R}$. The same statement
holds more generally for $P,Q\in\mathcal{L}_{+}(\mathcal{H})$. The fidelity is
invariant with respect to isometries and monotone non-decreasing with respect
to channels. The sine distance or $C$-distance between two quantum states
$\rho,\sigma\in\mathcal{D}(\mathcal{H})$ was defined as
\begin{equation}
C(\rho,\sigma)\equiv\sqrt{1-F(\rho,\sigma)}%
\end{equation}
and proven to be a metric in \cite{R02,R03,GLN04,R06}. It was
later~\cite{TCR09} (under the name \textquotedblleft purified
distance\textquotedblright) shown to be a metric on subnormalized states
$\rho,\sigma\in\mathcal{D}_{\leq}(\mathcal{H})$ via the embedding
\begin{equation}
P(\rho,\sigma)\equiv C(\rho\oplus\left[  1-\operatorname{Tr}\{\rho\}\right]
,\sigma\oplus\left[  1-\operatorname{Tr}\{\sigma\}\right]  )\,.
\label{eq:purified-distance}%
\end{equation}
The following inequality relates trace distance and purified distance:%
\begin{equation}
\frac{1}{2}\left\Vert \rho-\sigma\right\Vert _{1}\leq P(\rho,\sigma).
\label{eq:TD-to-PD}%
\end{equation}

\bigskip

\textbf{Relative entropies and variances.} The quantum relative entropy of two
states $\omega$ and $\tau$ is defined as \cite{U62}%
\begin{equation}
D(\omega\Vert\tau)\equiv\operatorname{Tr}\{\omega\lbrack\log_{2}\omega
-\log_{2}\tau]\}
\end{equation}
whenever $\operatorname{supp}(\omega)\subseteq\operatorname{supp}(\tau)$ and
it is equal to $+\infty$ otherwise. The quantum relative entropy variance is
defined as \cite{TH12,li12}%
\begin{equation}
V(\omega\Vert\tau)\equiv\operatorname{Tr}\{\omega\lbrack\log_{2}\omega
-\log_{2}\tau-D(\omega\Vert\tau)]^{2}\},
\end{equation}
whenever $\operatorname{supp}(\omega)\subseteq\operatorname{supp}(\tau)$. The
hypothesis testing relative entropy \cite{BD10,WR12}\ of states $\omega$ and
$\tau$ is defined as%
\begin{equation}
D_{H}^{\varepsilon}(\omega\Vert\tau)\equiv-\log_{2}\inf_{\Lambda}\left\{
\operatorname{Tr}\{\Lambda\tau\}:0\leq\Lambda\leq I\wedge\operatorname{Tr}%
\{\Lambda\omega\}\geq1-\varepsilon\right\}  .
\end{equation}
The max-relative entropy for states $\omega$ and $\tau$ is defined as
\cite{D09}%
\begin{equation}
D_{\max}(\omega\Vert\tau)\equiv\inf\left\{  \lambda\in\mathbb{R}:\omega
\leq2^{\lambda}\tau\right\}  .
\end{equation}
The smooth max-relative entropy for states $\omega$ and $\tau$ and a parameter
$\varepsilon\in(0,1)$ is defined as \cite{D09}%
\begin{equation}
D_{\max}^{\varepsilon}(\omega\Vert\tau)\equiv\inf\left\{  \lambda\in
\mathbb{R}:\widetilde{\omega}\leq2^{\lambda}\tau\wedge P(\omega,\widetilde
{\omega})\leq\varepsilon\right\}  .
\end{equation}
The following second-order expansions hold for $D_{H}^{\varepsilon}$ and
$D_{\max}^{\varepsilon}$ when evaluated for tensor-power states
\cite{TH12,li12}:%
\begin{align}
D_{H}^{\varepsilon}(\omega^{\otimes n}\Vert\tau^{\otimes n})  &
=nD(\omega\Vert\tau)+\sqrt{nV(\omega\Vert\tau)}\Phi^{-1}(\varepsilon)+O(\log
n),\label{eq:expand-1}\\
D_{\max}^{\sqrt{\varepsilon}}(\omega^{\otimes n}\Vert\tau^{\otimes n})  &
=nD(\omega\Vert\tau)-\sqrt{nV(\omega\Vert\tau)}\Phi^{-1}(\varepsilon)+O(\log
n). \label{eq:expand-2}%
\end{align}
The above expansion features the cumulative distribution function for a
standard normal random variable:%
\begin{equation}
\Phi(a)\equiv\frac{1}{\sqrt{2\pi}}\int_{-\infty}^{a}dx\,\exp\left(
-x^{2}/2\right)  , \label{eq:cumul-gauss}%
\end{equation}
and its inverse, defined as $\Phi^{-1}(\varepsilon)\equiv\sup\left\{
a\in\mathbb{R}\,|\,\Phi(a)\leq\varepsilon\right\}  $.

\bigskip

\textbf{Mutual informations and variances.} The quantum mutual information
$I(X;B)_{\rho}$ and information variance $V(X;B)_{\rho}$ of a bipartite state
$\rho_{XB}$ are defined as%
\begin{align}
I(X;B)_{\rho}  &  \equiv D(\rho_{XB}\Vert\rho_{X}\otimes\rho_{B}),\\
V(X;B)_{\rho}  &  \equiv V(\rho_{XB}\Vert\rho_{X}\otimes\rho_{B}).
\end{align}
In this paper, we are exclusively interested in the case in which system $X$
of $\rho_{XB}$ is classical, so that $\rho_{XB}$ can be written as%
\begin{equation}
\rho_{XB}=\sum_{x}p_{X}(x)|x\rangle\langle x|_{X}\otimes\rho_{B}^{x},
\end{equation}
where $p_{X}$ is a probability distribution, $\{|x\rangle_{X}\}_{x}$ is an
orthonormal basis, and $\{\rho_{B}^{x}\}_{x}$ is a set of quantum states. The
hypothesis testing mutual information is defined as follows for a bipartite
state $\rho_{XB}$ and a parameter $\varepsilon\in(0,1)$:%
\begin{equation}
I_{H}^{\varepsilon}(X;B)_{\rho}\equiv D_{H}^{\varepsilon}(\rho_{XB}\Vert
\rho_{X}\otimes\rho_{B}).
\end{equation}
From the smooth max-relative entropy, one can define a mutual information-like
quantity for a state $\theta_{AB}$ as follows:%
\begin{equation}
D_{\max}^{\varepsilon}(\theta_{AB}\Vert\theta_{A}\otimes\theta_{B}).
\label{eq:D_max-MI}%
\end{equation}
Note that we have the following expansions, as a direct consequence of
\eqref{eq:expand-1}--\eqref{eq:expand-2}\ and definitions:%
\begin{align}
I_{H}^{\varepsilon}(X^{n};B^{n})_{\rho^{\otimes n}}  &  =nI(X;B)_{\rho}%
+\sqrt{nV(X;B)_{\rho}}\Phi^{-1}(\varepsilon)+O(\log n),\label{eq:MI-expand}\\
D_{\max}^{\sqrt{\varepsilon}}(\rho_{XB}^{\otimes n}\Vert\rho_{X}^{\otimes
n}\otimes\rho_{B}^{\otimes n})  &  =nI(X;B)_{\rho}-\sqrt{nV(X;B)_{\rho}}%
\Phi^{-1}(\varepsilon)+O(\log n). \label{eq:MI-expand-2}%
\end{align}
Another quantity, related to that in \eqref{eq:D_max-MI}, is as follows
\cite{AJW17}:%
\begin{equation}
\widetilde{I}_{\max}^{\varepsilon}(B;A)_{\theta}\equiv\inf_{\theta
_{AB}^{\prime}\ :\ P(\theta_{AB}^{\prime},\theta_{AB})\leq\varepsilon}D_{\max
}(\theta_{AB}^{\prime}\Vert\theta_{A}\otimes\theta_{B}^{\prime}).
\label{eq:tilde-I-max}%
\end{equation}
We recall a relation \cite[Lemma~1]{AJW17} between the quantities in
\eqref{eq:D_max-MI}\ and \eqref{eq:tilde-I-max}, giving a very slight
modification of it which will be useful for our purposes here:

\begin{lemma}
\label{lem:alt-dmax-to-dmax-smooth}For a state $\theta_{AB}$, $\varepsilon
\in(0,1)$, and $\gamma\in(0,\varepsilon)$, the following inequality holds%
\begin{equation}
\widetilde{I}_{\max}^{\varepsilon}(B;A)_{\theta}\leq D_{\max}^{\varepsilon
-\gamma}(\theta_{AB}\Vert\theta_{A}\otimes\theta_{B})+\log_{2}\!\left(
\frac{3}{\gamma^{2}}\right)  . \label{eq:alt-dmax-to-dmax}%
\end{equation}

\end{lemma}

\begin{proof}
To see this, recall \cite[Claim~2]{AJW17}: For states $\omega_{AB}$, $\xi_{A}%
$, $\kappa_{B}$, there exists a state $\overline{\omega}_{AB}$ such that
$P(\omega_{AB},\overline{\omega}_{AB})\leq\delta$ and%
\begin{equation}
D_{\max}(\overline{\omega}_{AB}\Vert\xi_{A}\otimes\overline{\omega}_{B})\leq
D_{\max}(\omega_{AB}\Vert\xi_{A}\otimes\kappa_{B})+\log_{2}\!\left(  \frac
{3}{\delta^{2}}\right)  . \label{eq:claim-2}%
\end{equation}
Let $\theta_{AB}^{\ast}$ denote the optimizer for $D_{\max}^{\varepsilon
-\gamma}(\theta_{AB}\Vert\theta_{A}\otimes\theta_{B})$. Then, in
\eqref{eq:claim-2}, taking $\omega_{AB}=\theta_{AB}^{\ast}$, $\xi_{A}%
=\theta_{A}$, $\kappa_{B}=\theta_{B}$, we find that there exists a state
$\overline{\theta}_{AB}$ such that $P(\theta_{AB}^{\ast},\overline{\theta
}_{AB})\leq\gamma$ and%
\begin{equation}
D_{\max}(\overline{\theta}_{AB}\Vert\theta_{A}\otimes\overline{\theta}%
_{B})\leq D_{\max}^{\varepsilon-\gamma}(\theta_{AB}\Vert\theta_{A}%
\otimes\theta_{B})+\log_{2}\!\left(  \frac{3}{\gamma^{2}}\right)  .
\end{equation}
By the triangle inequality for the purified distance, we conclude that
$P(\theta_{AB},\overline{\theta}_{AB})\leq P(\theta_{AB},\theta_{AB}^{\ast
})+P(\theta_{AB}^{\ast},\overline{\theta}_{AB})\leq\left(  \varepsilon
-\gamma\right)  +\gamma=\varepsilon$. Since the quantity on the left-hand side
includes an optimization over all states $\theta_{AB}^{\prime}$ satisfying
$P(\theta_{AB}^{\prime},\theta_{AB})\leq\varepsilon$, we conclude the
inequality in \eqref{eq:alt-dmax-to-dmax}.
\end{proof}

\bigskip

\textbf{Hayashi--Nagaoka operator inequality.} A key tool in analyzing error
probabilities in communication protocols is the Hayashi--Nagaoka operator
inequality \cite{HN03}:\ given operators $S$ and $T$ such that $0\leq S\leq I$
and $T\geq0$, the following inequality holds for all $c>0$%
\begin{equation}
I-(S+T)^{-1/2}S(S+T)^{-1/2}\leq(1+c)(I-S)+(2+c+c^{-1})T. \label{eq:HN-ineq}%
\end{equation}

\bigskip

\textbf{Convex-split lemma.} The convex-split lemma from \cite{ADJ17} has been
a key tool used in recent developments in quantum information theory
\cite{AJW17,ADJ17}. We now state a variant of the convex-split lemma, which is
helpful for obtaining one-shot bounds for privacy and an ensuing lower bound
on the second-order coding rate. Its proof closely follows proofs available in
\cite{AJW17,ADJ17} but has some slight differences. For completeness,
Appendix~\ref{app:convex-split} contains a proof of
Lemma~\ref{thm:convex-split}.

\begin{lemma}
[Convex split]\label{thm:convex-split}Let $\rho_{AB}$ be a state, and let
$\tau_{A_{1}\cdots A_{K}B}$ be the following state:%
\begin{equation}
\tau_{A_{1}\cdots A_{K}B}\equiv\frac{1}{K}\sum_{k=1}^{K}\rho_{A_{1}}%
\otimes\cdots\otimes\rho_{A_{k-1}}\otimes\rho_{A_{k}B}\otimes\rho_{A_{k+1}%
}\otimes\cdots\otimes\rho_{A_{K}}.
\end{equation}
Let $\varepsilon\in(0,1)$ and $\eta\in(0,\sqrt{\varepsilon})$. If%
\begin{equation}
\log_{2}K=\widetilde{I}_{\max}^{\sqrt{\varepsilon}-\eta}(B;A)_{\rho}+2\log
_{2}\!\left(  \frac{1}{\eta}\right)  ,
\end{equation}
then%
\begin{equation}
P(\tau_{A_{1}\cdots A_{K}B},\rho_{A_{1}}\otimes\cdots\otimes\rho_{A_{K}%
}\otimes\widetilde{\rho}_{B})\leq\sqrt{\varepsilon},
\label{eq:perf-error-convex-split}%
\end{equation}
for some state $\widetilde{\rho}_{B}$ such that $P(\rho_{B},\widetilde{\rho
}_{B})\leq\sqrt{\varepsilon}-\eta$.
\end{lemma}

\section{Public classical communication\label{sec:public-comm}}

\subsection{Definition of the one-shot classical capacity}

We begin by defining the $\varepsilon$-one-shot classical capacity of a cq
channel%
\begin{equation}
x\rightarrow\rho_{B}^{x}. \label{eq:cq-channel}%
\end{equation}
We can write the classical--quantum channel in fully quantum form as the
following quantum channel:%
\begin{equation}
\mathcal{N}_{X^{\prime}\rightarrow B}(\sigma_{X^{\prime}})=\sum_{x}\langle
x|_{X^{\prime}}\sigma_{X^{\prime}}|x\rangle_{X^{\prime}}\rho_{B}^{x},
\end{equation}
where $\{|x\rangle_{X^{\prime}}\}_{x}$ is some orthonormal basis. Let
$M\in\mathbb{N}$ and $\varepsilon\in(0,1)$. An $(M,\varepsilon)$ classical
communication code consists of a collection of probability distributions
$\{p_{X|M}(x|m)\}_{m=1}^{M}$ (one for each message $m$) and a decoding
positive operator-valued measure (POVM) $\{\Lambda_{B}^{m}\}_{m=1}^{M}%
$,\footnote{We could allow for a decoding POVM\ to be $\{\Lambda_{B}%
^{m}\}_{m=0}^{M}$, consisting of an extra operator $\Lambda_{B}^{0}=I_{B}%
-\sum_{m=1}^{M}\Lambda_{B}^{m}$, if needed.} such that%
\begin{equation}
\frac{1}{M}\sum_{m=1}^{M}\operatorname{Tr}\{(I_{B}-\Lambda_{B}^{m})\rho
_{B}^{m}\}=\frac{1}{M}\sum_{m=1}^{M}\frac{1}{2}\left\Vert \mathcal{M}%
_{B\rightarrow\hat{M}}(\rho_{B}^{m})-|m\rangle\langle m|_{\hat{M}}\right\Vert
_{1}\leq\varepsilon. \label{eq:decoding-error}%
\end{equation}
We refer to the left-hand side of the above inequality as the \textit{decoding
error}. In the above, $\{|m\rangle_{\hat{M}}\}_{m=1}^{M}$ is an orthonormal
basis, we define the state $\rho_{B}^{m}$ as%
\begin{equation}
\rho_{B}^{m}=\sum_{x}p_{X|M}(x|m)\rho_{B}^{x},
\end{equation}
and the measurement channel $\mathcal{M}_{B\rightarrow\hat{M}}$\ as%
\begin{equation}
\mathcal{M}_{B\rightarrow\hat{M}}(\omega_{B})\equiv\sum_{m}\operatorname{Tr}%
\{\Lambda_{B}^{m}\omega_{B}\}|m\rangle\langle m|_{\hat{M}}.
\end{equation}
The equality in \eqref{eq:decoding-error} follows by direct calculation:%
\begin{align}
&  \left\Vert \mathcal{M}_{B\rightarrow\hat{M}}(\rho_{B}^{m})-|m\rangle\langle
m|_{\hat{M}}\right\Vert _{1}\nonumber\\
&  =\left\Vert \sum_{m^{\prime}}\operatorname{Tr}\{\Lambda_{B}^{m^{\prime}%
}\rho_{B}^{m}\}|m^{\prime}\rangle\langle m^{\prime}|_{\hat{M}}-|m\rangle
\langle m|_{\hat{M}}\right\Vert _{1}\\
&  =\left\Vert \sum_{m^{\prime}\neq m}\operatorname{Tr}\{\Lambda
_{B}^{m^{\prime}}\rho_{B}^{m}\}|m^{\prime}\rangle\langle m^{\prime}|_{\hat{M}%
}-(1-\operatorname{Tr}\{\Lambda_{B}^{m}\rho_{B}^{m}\})|m\rangle\langle
m|_{\hat{M}}\right\Vert _{1}\\
&  =\sum_{m^{\prime}\neq m}\operatorname{Tr}\{\Lambda_{B}^{m^{\prime}}\rho
_{B}^{m}\}+(1-\operatorname{Tr}\{\Lambda_{B}^{m}\rho_{B}^{m}\})\\
&  =2\operatorname{Tr}\{(I_{B}-\Lambda_{B}^{m})\rho_{B}^{m}\}.
\end{align}
For a given channel $\mathcal{N}_{X^{\prime}\rightarrow B}$ and $\varepsilon$,
the one-shot classical capacity is equal to $\log_{2}M_{\operatorname{pub}%
}^{\ast}(\varepsilon)$, where $M_{\operatorname{pub}}^{\ast}(\varepsilon)$ is
the largest $M$ such that \eqref{eq:decoding-error} can be satisfied for a
fixed $\varepsilon$.

One can allow for shared randomness between Alice and Bob before communication
begins, in which case one obtains the one-shot shared randomness assisted
capacity of a cq channel.

\subsection{Lower bound on the one-shot, randomness-assisted classical
capacity}

\label{sec:random-assisted-public-classical}We first consider a one-shot
protocol for randomness assisted, public classical communication in which the
goal is for Alice to use the classical-input quantum-output (cq) channel in
\eqref{eq:cq-channel}\ once to send one of $M$ messages with error probability
no larger than $\varepsilon\in(0,1)$. The next section shows how to
derandomize such that the shared randomness is not needed.

\textit{The main result of this section is that}%
\begin{equation}
\mathit{\ }I_{H}^{\varepsilon-\eta}(X;B)_{\rho}-\log_{2}(4\varepsilon/\eta
^{2})\mathit{,\ }%
\end{equation}
\textit{for all }$\eta\in(0,\varepsilon)$\textit{, is a lower bound on the
}$\varepsilon$\textit{-one-shot randomness-assisted, classical capacity of the
cq channel in \eqref{eq:cq-channel}.} Although this result is already known
from \cite{WR12}, the development in this section is an important building
block for the wiretap channel result in Section~\ref{sec:priv-cap-lower}, and
so we go through it in full detail for the sake of completeness. Also, the
approach given here uses position-based decoding for the cq channel.

Fix a probability distribution $p_{X}$ over the channel input alphabet.
Consider the following classical--classical state:%
\begin{equation}
\rho_{XX^{\prime}}\equiv\sum_{x}p_{X}(x)|x\rangle\langle x|_{X}\otimes
|x\rangle\langle x|_{X^{\prime}},
\end{equation}
which we can think of as representing shared randomness. Let $\rho_{XB}$
denote the following state, which results from sending the $X^{\prime}$ system
of $\rho_{XX^{\prime}}$ through the channel $\mathcal{N}_{X^{\prime
}\rightarrow B}$:%
\begin{equation}
\rho_{XB}\equiv\mathcal{N}_{X^{\prime}\rightarrow B}(\rho_{XX^{\prime}}%
)=\sum_{x}p_{X}(x)|x\rangle\langle x|_{X}\otimes\rho_{B}^{x}.
\end{equation}

The coding scheme works as follows. Let Alice and Bob share $M$ copies of the
state $\rho_{XX^{\prime}}$, so that their shared state is%
\begin{equation}
\rho_{X^{M}X^{\prime M}}\equiv\rho_{X_{1}X_{1}^{\prime}}\otimes\cdots
\otimes\rho_{X_{M}X_{M}^{\prime}}=\rho_{XX^{\prime}}^{\otimes M}.
\end{equation}
Alice has the systems labeled by $X^{\prime}$, and Bob has the systems labeled
by $X$. If Alice would like to communicate message $m$ to Bob, then she simply
sends system $X_{m}^{\prime}$ over the classical--quantum channel. In such a
case, the reduced state for Bob is as follows:%
\begin{equation}
\rho_{X^{M}B}^{m}\equiv\rho_{X_{1}}\otimes\cdots\otimes\rho_{X_{m-1}}%
\otimes\rho_{X_{m+1}}\otimes\cdots\otimes\rho_{X_{M}}\otimes\rho_{X_{m}B}.
\end{equation}
Observe that each state $\rho_{X^{M}B}^{m}$ is related to the first one
$\rho_{X^{M}B}^{1}$ by a permutation $\pi(m)$ of the $X^{M}$ systems:%
\begin{equation}
W_{X^{M}}^{\pi(m)}\rho_{X^{M}B}^{1}W_{X^{M}}^{\pi(m)\dag}=\rho_{X^{M}B}^{m},
\label{eq:perm-inv-state}%
\end{equation}
where $W_{X^{M}}^{\pi(m)}$ is a unitary representation of the permutation
$\pi(m)$.

If Bob has a way of distinguishing the joint state $\rho_{XB}$ from the
product state $\rho_{X}\otimes\rho_{B}$, then with high probability, he will
be able to figure out which message $m$ was communicated. Let $T_{XB}$ denote
a test (measurement operator) satisfying $0\leq T_{XB}\leq I_{XB}$, which we
think of as identifying $\rho_{XB}$ with high probability ($\geq1-\varepsilon
$) and for which the complementary operator $I_{XB}-T_{XB}$\ identifies
$\rho_{X}\otimes\rho_{B}$ with the highest probability subject to the
constraint $\operatorname{Tr}\{T_{XB}\rho_{XB}\}\geq1-\varepsilon$. From such
a test, we form the following measurement operator:%
\begin{equation}
\Gamma_{X^{M}B}^{m}\equiv T_{X_{m}B}\otimes I_{X_{1}}\otimes\cdots\otimes
I_{X_{m-1}}\otimes I_{X_{m+1}}\otimes\cdots\otimes I_{X_{M}},
\end{equation}
which we think of as a test to figure out whether the reduced state on systems
$X_{m}B$ is $\rho_{XB}$ or $\rho_{X}\otimes\rho_{B}$. Observe that each
message operator $\Gamma_{X^{M}B}^{m}$ is related to the first one
$\Gamma_{X^{M}B}^{1}$ by a permutation $\pi(m)$ of the $X^{M}$ systems:%
\begin{equation}
W_{X^{M}}^{\pi(m)}\Gamma_{X^{M}B}^{1}W_{X^{M}}^{\pi(m)\dag}=\Gamma_{X^{M}%
B}^{m}.
\end{equation}
If message $m$ is transmitted and the measurement operator $\Gamma_{X^{M}%
B}^{m}$ acts, then the probability of it accepting is%
\begin{equation}
\operatorname{Tr}\{\Gamma_{X^{M}B}^{m}\rho_{X^{M}B}^{m}\}=\operatorname{Tr}%
\{T_{XB}\rho_{XB}\}.
\end{equation}
If however the measurement operator $\Gamma_{X^{M}B}^{m^{\prime}}$ acts, where
$m^{\prime}\neq m$, then the probability of it accepting is%
\begin{equation}
\operatorname{Tr}\{\Gamma_{X^{M}B}^{m^{\prime}}\rho_{X^{M}B}^{m}%
\}=\operatorname{Tr}\{T_{XB}[\rho_{X}\otimes\rho_{B}]\}.
\end{equation}
From these measurement operators, we then form a square-root measurement as
follows:%
\begin{equation}
\Lambda_{X^{M}B}^{m}\equiv\left(  \sum_{m^{\prime}=1}^{M}\Gamma_{X^{M}%
B}^{m^{\prime}}\right)  ^{-1/2}\Gamma_{X^{M}B}^{m}\left(  \sum_{m^{\prime}%
=1}^{M}\Gamma_{X^{M}B}^{m^{\prime}}\right)  ^{-1/2}.
\end{equation}
Again, each message operator $\Lambda_{X^{M}B}^{m}$ is related to the first
one $\Lambda_{X^{M}B}^{1}$ by a permutation of the $X^{M}$ systems:%
\begin{align}
&  W_{X^{M}}^{\pi(m)}\Lambda_{X^{M}B}^{1}W_{X^{M}}^{\pi(m)\dag}\nonumber\\
&  =W_{X^{M}}^{\pi(m)}\left(  \sum_{m^{\prime}=1}^{M}\Gamma_{X^{M}%
B}^{m^{\prime}}\right)  ^{-1/2}\Gamma_{X^{M}B}^{1}\left(  \sum_{m^{\prime}%
=1}^{M}\Gamma_{X^{M}B}^{m^{\prime}}\right)  ^{-1/2}W_{X^{M}}^{\pi(m)\dag
}\label{eq:perm-inv-meas-1}\\
&  =W_{X^{M}}^{\pi(m)}\left(  \sum_{m^{\prime}=1}^{M}\Gamma_{X^{M}%
B}^{m^{\prime}}\right)  ^{-1/2}W_{X^{M}}^{\pi(m)\dag}W_{X^{M}}^{\pi(m)}%
\Gamma_{X^{M}B}^{1}W_{X^{M}}^{\pi(m)\dag}W_{X^{M}}^{\pi(m)}\left(
\sum_{m^{\prime}=1}^{M}\Gamma_{X^{M}B}^{m^{\prime}}\right)  ^{-1/2}W_{X^{M}%
}^{\pi(m)\dag}\\
&  =\left(  \sum_{m^{\prime}=1}^{M}W_{X^{M}}^{\pi(m)}\Gamma_{X^{M}%
B}^{m^{\prime}}W_{X^{M}}^{\pi(m)\dag}\right)  ^{-1/2}\Gamma_{X^{M}B}%
^{m}\left(  \sum_{m^{\prime}=1}^{M}W_{X^{M}}^{\pi(m)}\Gamma_{X^{M}%
B}^{m^{\prime}}W_{X^{M}}^{\pi(m)\dag}\right)  ^{-1/2}\\
&  =\left(  \sum_{m^{\prime}=1}^{M}\Gamma_{X^{M}B}^{m^{\prime}}\right)
^{-1/2}\Gamma_{X^{M}B}^{m}\left(  \sum_{m^{\prime}=1}^{M}\Gamma_{X^{M}%
B}^{m^{\prime}}\right)  ^{-1/2}. \label{eq:perm-inv-meas-4}%
\end{align}

This is called the position-based decoder and was analyzed in \cite{AJW17}%
\ for the case of entanglement-assisted communication. The error probability
under this coding scheme is as follows for each message $m$:%
\begin{equation}
\operatorname{Tr}\{(I_{R^{M}B}-\Lambda_{R^{M}B}^{m})\rho_{R^{M}B}^{m}\}.
\end{equation}
The error probability is in fact the same for each message, due to the
observations in \eqref{eq:perm-inv-state}\ and
\eqref{eq:perm-inv-meas-1}--\eqref{eq:perm-inv-meas-4}:%
\begin{align}
\operatorname{Tr}\{(I_{R^{M}B}-\Lambda_{R^{M}B}^{1})\rho_{R^{M}B}^{1}\}  &
=\operatorname{Tr}\{(I_{R^{M}B}-\Lambda_{R^{M}B}^{1})W_{X^{M}}^{\pi(m)\dag
}W_{X^{M}}^{\pi(m)}\rho_{R^{M}B}^{1}W_{X^{M}}^{\pi(m)\dag}W_{X^{M}}^{\pi
(m)}\}\\
&  =\operatorname{Tr}\{(I_{R^{M}B}-W_{X^{M}}^{\pi(m)}\Lambda_{R^{M}B}%
^{1}W_{X^{M}}^{\pi(m)\dag})W_{X^{M}}^{\pi(m)}\rho_{R^{M}B}^{1}W_{X^{M}}%
^{\pi(m)\dag}\}\\
&  =\operatorname{Tr}\{(I_{R^{M}B}-\Lambda_{R^{M}B}^{m})\rho_{R^{M}B}^{m}\}.
\end{align}
So let us analyze the error probability for the first message $m=1$. Applying
the Hayashi-Nagaoka operator inequality in \eqref{eq:HN-ineq}, with
$S=\Gamma_{X^{M}B}^{1}$, $T=\sum_{m^{\prime}\neq1}\Gamma_{X^{M}B}^{m^{\prime}%
}$, $c_{\operatorname{I}}\equiv1+c$, and $c_{\operatorname{II}}=2+c+c^{-1}$
for $c>0$, we find that this error probability can be bounded from above as%
\begin{align}
&  \operatorname{Tr}\{(I_{R^{M}B}-\Lambda_{R^{M}B}^{1})\rho_{R^{M}B}%
^{1}\}\nonumber\\
&  \leq c_{\operatorname{I}}\operatorname{Tr}\{(I_{X^{M}B}-\Gamma_{X^{M}B}%
^{1})\rho_{X^{M}B}^{1}\}+c_{\operatorname{II}}\sum_{m^{\prime}\neq
1}\operatorname{Tr}\{\Gamma_{X^{M}B}^{m^{\prime}}\rho_{X^{M}B}^{1}\}\\
&  =c_{\operatorname{I}}\operatorname{Tr}\{(I_{XB}-T_{XB})\rho_{XB}%
\}+c_{\operatorname{II}}\sum_{m^{\prime}\neq1}\operatorname{Tr}\{T_{RB}\left[
\rho_{X}\otimes\rho_{B}\right]  \}\\
&  =c_{\operatorname{I}}\operatorname{Tr}\{(I_{XB}-T_{XB})\rho_{XB}%
\}+c_{\operatorname{II}}(M-1)\operatorname{Tr}\{T_{XB}\left[  \rho_{X}%
\otimes\rho_{B}\right]  \}. \label{eq:EA-err-prob}%
\end{align}

Consider the hypothesis testing mutual information:%
\begin{equation}
I_{H}^{\varepsilon}(X;B)_{\rho}\equiv D_{H}^{\varepsilon}(\rho_{XB}\Vert
\rho_{X}\otimes\rho_{B}),
\end{equation}
where%
\begin{equation}
D_{H}^{\varepsilon}(\rho\Vert\sigma)\equiv-\log_{2}\inf_{\Lambda}\left\{
\operatorname{Tr}\{\Lambda\sigma\}:0\leq\Lambda\leq I\wedge\operatorname{Tr}%
\{\Lambda\rho\}\geq1-\varepsilon\right\}  .
\end{equation}
Take the test $T_{XB}$ in Bob's decoder to be $\Upsilon_{XB}^{\ast}$, where
$\Upsilon_{XB}^{\ast}$ is the optimal measurement operator\ for $I_{H}%
^{\varepsilon-\eta}(X;B)_{\rho}$ for $\eta\in(0,\varepsilon)$. Then the error
probability is bounded as%
\begin{align}
&  \operatorname{Tr}\{(I_{X^{M}B}-\Lambda_{X^{M}B}^{1})\rho_{X^{M}B}%
^{1}\}\nonumber\\
&  \leq c_{\operatorname{I}}\operatorname{Tr}\{(I_{XB}-\Upsilon_{XB}^{\ast
})\rho_{XB}\}+c_{\operatorname{II}}M\operatorname{Tr}\{\Upsilon_{XB}^{\ast
}\left[  \rho_{X}\otimes\rho_{B}\right]  \}\\
&  \leq c_{\operatorname{I}}\left(  \varepsilon-\eta\right)
+c_{\operatorname{II}}M2^{-I_{H}^{\varepsilon-\eta}(X;B)_{\rho}}.
\end{align}
Now pick $c=\eta/(2\varepsilon-\eta)$ and we get that the last line above
$=\varepsilon$, for%
\begin{equation}
\log_{2}M=I_{H}^{\varepsilon-\eta}(X;B)_{\rho}-\log_{2}(4\varepsilon/\eta
^{2}).
\end{equation}
Indeed, consider that we would like to have $c$ such that%
\begin{equation}
\varepsilon=c_{\operatorname{I}}\left(  \varepsilon-\eta\right)
+c_{\operatorname{II}}M2^{-I_{H}^{\varepsilon-\eta}(X;B)_{\rho}}.
\end{equation}
Rewriting this, we find that $M$ should satisfy%
\begin{equation}
\log_{2}M=I_{H}^{\varepsilon-\eta}(X;B)_{\rho}+\log_{2}\!\left(
\frac{\varepsilon-c_{\operatorname{I}}\left(  \varepsilon-\eta\right)
}{c_{\operatorname{II}}}\right)  .
\end{equation}
Picking $c=\eta/(2\varepsilon-\eta)$ then implies (after some algebra)\ that%
\begin{equation}
\frac{\varepsilon-c_{\operatorname{I}}\left(  \varepsilon-\eta\right)
}{c_{\operatorname{II}}}=\frac{\eta^{2}}{4\varepsilon}.
\end{equation}

So the quantity $I_{H}^{\varepsilon-\eta}(X;B)_{\rho}-\log_{2}(4\varepsilon
/\eta^{2})$\ represents a lower bound on the $\varepsilon$-one-shot
randomness-assisted, classical capacity of the cq channel in
\eqref{eq:cq-channel}. The bound holds for both average error probability and
maximal error probability, and this coincidence is due to the protocol having
the assistance of shared randomness.

\subsection{Lower bound on the one-shot classical capacity}

\label{sec:derandomize-public-classical-code}Now I show how to derandomize the
above randomness-assisted code. The main result of this section is the
following lower bound on the $\varepsilon$-one shot classical capacity of the
cq channel in \eqref{eq:cq-channel}, holding for\ all\ $\eta\in(0,\varepsilon
)$:%
\begin{equation}
\log_{2}M_{\operatorname{pub}}^{\ast}(\varepsilon)\geq I_{H}^{\varepsilon
-\eta}(X;B)_{\rho}-\log_{2}(4\varepsilon/\eta^{2}).
\end{equation}
Again, note that although this result is already known from \cite{WR12}, the
development in this section is an important building block for the wiretap
channel result in Section~\ref{sec:priv-cap-lower}. As stated previously, the
approach given here uses position-based decoding for the cq channel.

By the reasoning from the previous section, we have the following bound on the
average error probability for a randomness-assisted code:%
\begin{equation}
\frac{1}{M}\sum_{m=1}^{M}\operatorname{Tr}\{(I_{X^{M}B}-\Lambda_{X^{M}B}%
^{m})\rho_{X^{M}B}^{m}\}\leq\varepsilon, \label{eq:avg-error-bound}%
\end{equation}
if%
\begin{equation}
\log_{2}M=I_{H}^{\varepsilon-\eta}(X;B)_{\rho}-\log_{2}(4\varepsilon/\eta
^{2}).
\end{equation}
So let us analyze the expression $\operatorname{Tr}\{(I_{X^{M}B}%
-\Lambda_{X^{M}B}^{m})\rho_{X^{M}B}^{m}\}$. By definition, it follows that%
\begin{equation}
\rho_{X^{M}B}^{m}=\sum_{x_{1},\ldots,x_{M}}p_{X}(x_{1})\cdots p_{X}%
(x_{M})|x_{1},\ldots,x_{M}\rangle\langle x_{1},\ldots,x_{M}|_{X_{1}\cdots
X_{M}}\otimes\rho_{B}^{x_{m}}. \label{eq:unraveled-state-m}%
\end{equation}
Also, recall that $\Upsilon_{XB}^{\ast}$ is optimal for $I_{H}^{\varepsilon
-\eta}(X;B)_{\rho}$, which implies that%
\begin{align}
\operatorname{Tr}\{\Upsilon_{XB}^{\ast}\rho_{XB}\}  &  \geq1-\left(
\varepsilon-\eta\right)  ,\\
\operatorname{Tr}\{\Upsilon_{XB}^{\ast}\left[  \rho_{X}\otimes\rho_{B}\right]
\}  &  =2^{-I_{H}^{\varepsilon-\eta}(X;B)_{\rho}}.
\end{align}
But consider that%
\begin{align}
\operatorname{Tr}\{\Upsilon_{XB}^{\ast}\rho_{XB}\}  &  =\operatorname{Tr}%
\left\{  \Upsilon_{XB}^{\ast}\sum_{x}p_{X}(x)|x\rangle\langle x|_{X}%
\otimes\rho_{B}^{x}\right\} \\
&  =\sum_{x}p_{X}(x)\operatorname{Tr}\left\{  \langle x|_{X}\Upsilon
_{XB}^{\ast}|x\rangle_{X}\rho_{B}^{x}\right\} \\
&  =\sum_{x}p_{X}(x)\operatorname{Tr}\left\{  Q_{B}^{x}\rho_{B}^{x}\right\}  ,
\end{align}
where we define%
\begin{equation}
Q_{B}^{x}\equiv\langle x|_{X}\Upsilon_{XB}^{\ast}|x\rangle_{X}.
\label{eq:optimal-Q-x}%
\end{equation}
Similarly, we have that%
\begin{align}
\operatorname{Tr}\{\Upsilon_{XB}^{\ast}\left[  \rho_{X}\otimes\rho_{B}\right]
\}  &  =\operatorname{Tr}\left\{  \Upsilon_{XB}^{\ast}\sum_{x}p_{X}%
(x)|x\rangle\langle x|_{X}\otimes\rho_{B}\right\} \\
&  =\sum_{x}p_{X}(x)\operatorname{Tr}\left\{  \langle x|_{X}\Upsilon
_{XB}^{\ast}|x\rangle_{X}\rho_{B}\right\} \\
&  =\sum_{x}p_{X}(x)\operatorname{Tr}\left\{  Q_{B}^{x}\rho_{B}\right\}  .
\end{align}
This demonstrates that it suffices to take the optimal measurement operator
$\Upsilon_{XB}^{\ast}$ to be $\sum_{x}|x\rangle\langle x|_{X}\otimes Q_{B}%
^{x}$, with $Q_{B}^{x}$ defined as in \eqref{eq:optimal-Q-x}, and this will
achieve the same optimal value as $\Upsilon_{XB}^{\ast}$ does.

Taking $\Upsilon_{XB}^{\ast}$ as such, now consider that%
\begin{align}
\Gamma_{X^{M}B}^{m}  &  =\Upsilon_{X_{m}B}^{\ast}\otimes I_{X_{1}}%
\otimes\cdots\otimes I_{X_{m-1}}\otimes I_{X_{m+1}}\otimes\cdots\otimes
I_{X_{M}}\\
&  =\sum_{x_{m}}|x_{m}\rangle\langle x_{m}|_{X_{m}}\otimes Q_{B}^{x_{m}%
}\otimes I_{X_{1}}\otimes\cdots\otimes I_{X_{m-1}}\otimes I_{X_{m+1}}%
\otimes\cdots\otimes I_{X_{M}}\\
&  =\sum_{x_{1},\ldots,x_{M}}|x_{1},\ldots,x_{M}\rangle\langle x_{1}%
,\ldots,x_{M}|_{X^{M}}\otimes Q_{B}^{x_{m}}.
\end{align}
Then this implies that%
\begin{align}
\left(  \sum_{m^{\prime}=1}^{M}\Gamma_{X^{M}B}^{m^{\prime}}\right)  ^{-1/2}
&  =\left(  \sum_{m^{\prime}=1}^{M}\sum_{x_{1},\ldots,x_{M}}|x_{1}%
,\ldots,x_{M}\rangle\langle x_{1},\ldots,x_{M}|_{X^{M}}\otimes Q_{B}%
^{x_{m^{\prime}}}\right)  ^{-1/2}\\
&  =\left(  \sum_{x_{1},\ldots,x_{M}}|x_{1},\ldots,x_{M}\rangle\langle
x_{1},\ldots,x_{M}|_{X^{M}}\otimes\sum_{m^{\prime}=1}^{M}Q_{B}^{x_{m^{\prime}%
}}\right)  ^{-1/2}\\
&  =\sum_{x_{1},\ldots,x_{M}}|x_{1},\ldots,x_{M}\rangle\langle x_{1}%
,\ldots,x_{M}|_{X^{M}}\otimes\left(  \sum_{m^{\prime}=1}^{M}Q_{B}%
^{x_{m^{\prime}}}\right)  ^{-1/2},
\end{align}
so that%
\begin{align}
\Lambda_{X^{M}B}^{m}  &  =\left(  \sum_{m^{\prime}=1}^{M}\Gamma_{X^{M}%
B}^{m^{\prime}}\right)  ^{-1/2}\Gamma_{X^{M}B}^{m}\left(  \sum_{m^{\prime}%
=1}^{M}\Gamma_{X^{M}B}^{m^{\prime}}\right)  ^{-1/2}\\
&  =\sum_{x_{1},\ldots,x_{M}}|x_{1},\ldots,x_{M}\rangle\langle x_{1}%
,\ldots,x_{M}|_{X^{M}}\otimes\Omega_{B}^{x_{m}}, \label{eq:unraveled-meas-m}%
\end{align}
where%
\begin{equation}
\Omega_{B}^{x_{m}}\equiv\left(  \sum_{m^{\prime}=1}^{M}Q_{B}^{x_{m^{\prime}}%
}\right)  ^{-1/2}Q_{B}^{x_{m}}\left(  \sum_{m^{\prime}=1}^{M}Q_{B}%
^{x_{m^{\prime}}}\right)  ^{-1/2}.
\end{equation}
Observe that $\{\Omega_{B}^{x_{m}}\}_{m=1}^{M}$ is a POVM on the support of
$\sum_{m^{\prime}=1}^{M}Q_{B}^{x_{m^{\prime}}}$ and can be completed to a
POVM\ on the full space by adding $\Omega_{B}^{x_{0}}\equiv I_{B}%
-\sum_{m^{\prime}=1}^{M}Q_{B}^{x_{m^{\prime}}}$. By employing
\eqref{eq:unraveled-state-m} and \eqref{eq:unraveled-meas-m}, we find that%
\begin{equation}
\operatorname{Tr}\{(I_{X^{M}B}-\Lambda_{X^{M}B}^{m})\rho_{X^{M}B}^{m}%
\}=\sum_{x_{1},\ldots,x_{M}}p_{X}(x_{1})\cdots p_{X}(x_{M})\operatorname{Tr}%
\{(I_{B}-\Omega_{B}^{x_{m}})\rho_{B}^{x_{m}}\},
\end{equation}
so that the average error probability is as follows:%
\begin{align}
&  \frac{1}{M}\sum_{m=1}^{M}\operatorname{Tr}\{(I_{X^{M}B}-\Lambda_{X^{M}%
B}^{m})\rho_{X^{M}B}^{m}\}\nonumber\\
&  =\frac{1}{M}\sum_{m=1}^{M}\sum_{x_{1},\ldots,x_{M}}p_{X}(x_{1})\cdots
p_{X}(x_{M})\operatorname{Tr}\{(I_{B}-\Omega_{B}^{x_{m}})\rho_{B}^{x_{m}}\}\\
&  =\sum_{x_{1},\ldots,x_{M}}p_{X}(x_{1})\cdots p_{X}(x_{M})\left[  \frac
{1}{M}\sum_{m=1}^{M}\operatorname{Tr}\{(I_{B}-\Omega_{B}^{x_{m}})\rho
_{B}^{x_{m}}\}\right]  .
\end{align}
The last line above is the same as the usual \textquotedblleft Shannon
trick\textquotedblright\ of exchanging the average over the messages with the
expectation over a random choice of code. By employing the bound in
\eqref{eq:avg-error-bound}, we find that%
\begin{equation}
\sum_{x_{1},\ldots,x_{M}}p_{X}(x_{1})\cdots p_{X}(x_{M})\left[  \frac{1}%
{M}\sum_{m=1}^{M}\operatorname{Tr}\{(I_{B}-\Omega_{B}^{x_{m}})\rho_{B}^{x_{m}%
}\}\right]  \leq\varepsilon.
\end{equation}
Then there exists a particular set of values of $x_{1}$, \ldots, $x_{M}$ such
that%
\begin{equation}
\frac{1}{M}\sum_{m=1}^{M}\operatorname{Tr}\{(I_{B}-\Omega_{B}^{x_{m}})\rho
_{B}^{x_{m}}\}\leq\varepsilon.
\end{equation}
This sequence $x_{1}$, \ldots, $x_{M}$ constitutes the codewords and
$\{\Omega_{B}^{x_{m}}\}_{m=1}^{M}$ is a corresponding POVM that can be used as
a decoder. The number of bits that the code can transmit is equal to $\log
_{2}M=I_{H}^{\varepsilon-\eta}(X;B)_{\rho}-\log_{2}(4\varepsilon/\eta^{2})$.
No shared randomness is required for this code (it is now derandomized).

\begin{remark}
To achieve maximal error probability $2\varepsilon$, one can remove the worst
half of the codewords, and then a lower bound on the achievable number of bits
is%
\begin{equation}
I_{H}^{\varepsilon-\eta}(X;B)_{\rho}-\log_{2}2-\log_{2}(4\varepsilon/\eta
^{2})=I_{H}^{\varepsilon-\eta}(X;B)_{\rho}-\log_{2}(8\varepsilon/\eta^{2}).
\end{equation}

\end{remark}

\section{Private classical communication\label{sec:priv-comm}}

\subsection{Definition of the one-shot private classical capacity}

Now suppose that Alice, Bob, and Eve are connected by a classical-input
quantum-quantum-output (cqq) channel of the following form:%
\begin{equation}
x\rightarrow\rho_{BE}^{x},
\end{equation}
where Bob has system $B$ and Eve system $E$. The fully quantum version of this
channel is as follows:%
\begin{equation}
\mathcal{N}_{X^{\prime}\rightarrow BE}(\sigma_{X^{\prime}})=\sum_{x}\langle
x|_{X^{\prime}}\sigma_{X^{\prime}}|x\rangle_{X^{\prime}}\rho_{BE}^{x},
\label{eq:wiretap-fully-quantum}%
\end{equation}
where $\{|x\rangle_{X^{\prime}}\}_{x}$ is some orthonormal basis.

We define the one-shot private classical capacity in the following way. Let
$M\in\mathbb{N}$ and $\varepsilon\in(0,1)$. An $(M,\varepsilon)$ private
communication code consists of a collection of probability distributions
$\{p_{X|M}(x|m)\}_{m=1}^{M}$ (one for each message $m$) and a decoding POVM
$\{\Lambda_{B}^{m}\}_{m=1}^{M}$, such that%
\begin{equation}
\frac{1}{M}\sum_{m=1}^{M}\frac{1}{2}\left\Vert \mathcal{M}_{B\rightarrow
\hat{M}}(\rho_{BE}^{m})-|m\rangle\langle m|_{\hat{M}}\otimes\sigma
_{E}\right\Vert _{1}\leq\varepsilon. \label{eq:avg-rel-sec}%
\end{equation}
We refer to the left-hand side of the above inequality as the \textit{privacy
error}. In the above, $\{|m\rangle_{\hat{M}}\}_{m=1}^{M}$ is an orthonormal
basis, the state $\sigma_{E}$ can be any state, we define the state $\rho
_{BE}^{m}$ as%
\begin{equation}
\rho_{BE}^{m}=\sum_{x}p_{X|M}(x|m)\rho_{BE}^{x},
\end{equation}
and the measurement channel $\mathcal{M}_{B\rightarrow\hat{M}}$\ as%
\begin{equation}
\mathcal{M}_{B\rightarrow\hat{M}}(\omega_{B})\equiv\sum_{m}\operatorname{Tr}%
\{\Lambda_{B}^{m}\omega_{B}\}|m\rangle\langle m|_{\hat{M}}.
\end{equation}
For a given channel $\mathcal{N}_{X^{\prime}\rightarrow BE}$ and $\varepsilon
$, the one-shot private classical capacity is equal to $\log_{2}%
M_{\operatorname{priv}}^{\ast}(\varepsilon)$, where $M_{\operatorname{priv}%
}^{\ast}(\varepsilon)$ is the largest $M$ such that \eqref{eq:avg-rel-sec} can
be satisfied for a fixed $\varepsilon$.

The condition in \eqref{eq:avg-rel-sec} combines the reliable decoding and
security conditions into a single average error criterion. We can see how it
represents a generalization of the error criterion in
\eqref{eq:decoding-error}, which was for public classical communication over a
cq channel. One could have a different definition of one-shot private
capacity, in which there are two separate criteria, but the approach above
will be beneficial for our purposes. In any case, a code satisfying
\eqref{eq:avg-rel-sec}~satisfies the two separate criteria as well, as is
easily seen by invoking the monotonicity of trace distance.\footnote{Indeed,
starting with \eqref{eq:avg-rel-sec}\ and applying monotonicity of trace
distance under partial trace of the $E$ system, we get that $\frac{1}{M}%
\sum_{m=1}^{M}\frac{1}{2}\left\Vert \mathcal{M}_{B\rightarrow\hat{M}}(\rho
_{B}^{m})-|m\rangle\langle m|_{\hat{M}}\right\Vert _{1}\leq\varepsilon$.
Recalling \eqref{eq:decoding-error}, we can interpret this as asserting that
the decoding error probability does not exceed $\varepsilon$. Doing the same
but considering a partial trace over the $B$ system implies that $\frac{1}%
{M}\sum_{m=1}^{M}\frac{1}{2}\left\Vert \rho_{E}^{m}-\sigma_{E}\right\Vert
_{1}\leq\varepsilon$, which is a security criterion. So we get that the
conventional two separate criteria are satisfied if a code satisfies the
single privacy error criterion in \eqref{eq:avg-rel-sec}.} Having a single
error criterion for private capacity is the same as the approach taken in
\cite{HHHO09} and \cite{WTB16}, and in the latter paper, it was shown that
notions of asymptotic private capacity are equivalent when using either a
single error criterion or two separate error criteria.

\subsection{Lower bound on the one-shot private classical
capacity\label{sec:priv-cap-lower}}

The main result of this section is the following lower bound on the
$\varepsilon$-one shot private capacity of a cq wiretap channel, holding for
all $\varepsilon_{1},\varepsilon_{2}\in(0,1)$, such that $\varepsilon
_{1}+\sqrt{\varepsilon_{2}}\in(0,1)$, and $\eta_{1}\in(0,\varepsilon_{1})$ and
$\eta_{2}\in(0,\sqrt{\varepsilon_{2}})$:%
\begin{equation}
\log_{2}M_{\operatorname{priv}}^{\ast}(\varepsilon_{1}+\sqrt{\varepsilon_{2}%
})\geq I_{H}^{\varepsilon_{1}-\eta_{1}}(X;B)_{\rho}-\widetilde{I}_{\max
}^{\sqrt{\varepsilon_{2}}-\eta_{2}}(E;X)_{\rho}-\log_{2}(4\varepsilon_{1}%
/\eta_{1}^{2})-2\log_{2}\!\left(  1/\eta_{2}\right)  .
\end{equation}

To begin with, we allow Alice, Bob, and Eve shared randomness of the following
form:%
\begin{equation}
\rho_{XX^{\prime}X^{\prime\prime}}\equiv\sum_{x}p_{X}(x)|x\rangle\langle
x|_{X}\otimes|x\rangle\langle x|_{X^{\prime}}\otimes|x\rangle\langle
x|_{X^{\prime\prime}}, \label{eq:randomness-state-ABE}%
\end{equation}
where Bob has the $X$ system, Alice the $X^{\prime}$ system, and Eve the
$X^{\prime\prime}$ system. It is natural here to let Eve share the randomness
as well, and this amounts to giving her knowledge of the code to be used. Let
$\rho_{XX^{\prime\prime}BE}$ denote the state resulting from sending the
$X^{\prime}$ system through the channel $\mathcal{N}_{X^{\prime}\rightarrow
BE}$ in\ \eqref{eq:wiretap-fully-quantum}:%
\begin{equation}
\rho_{XX^{\prime\prime}BE}\equiv\sum_{x}p_{X}(x)|x\rangle\langle x|_{X}%
\otimes\rho_{BE}^{x}\otimes|x\rangle\langle x|_{X^{\prime\prime}}.
\end{equation}

The coding scheme that Alice and Bob use is as follows. There is the message
$m\in\{1,\ldots,M\}$ and a local key $k\in\{1,\ldots,K\}$. The local key $k$
represents local, uniform randomness that Alice has, but which is not
accessible to Bob or Eve. We assume that Alice, Bob, and Eve share $MK$ copies
of the state in \eqref{eq:randomness-state-ABE} before communication begins,
and we denote this state as%
\begin{equation}
\rho_{X^{MK}X^{\prime MK}X^{\prime\prime MK}}=\rho_{X_{1,1}X_{1,1}^{\prime
}X_{1,1}^{\prime\prime}}\otimes\cdots\otimes\rho_{X_{M,K}X_{M,K}^{\prime
}X_{M,K}^{\prime\prime}}=\rho_{XX^{\prime}X^{\prime\prime}}^{\otimes MK}.
\end{equation}
To send the message $m$, Alice picks $k$ uniformly at random from the set
$\{1,\ldots,K\}$. She then sends the $(m,k)$th $X^{\prime}$ system through the
channel $\mathcal{N}_{X^{\prime}\rightarrow BE}$. Thus, when $m$ and $k$ are
chosen, the reduced state on Bob and Eve's systems is%
\begin{equation}
\rho_{X^{MK}X^{\prime\prime MK}BE}^{m,k}=\rho_{X_{1,1}X_{1,1}^{\prime\prime}%
}\otimes\cdots\otimes\rho_{X_{m,k-1}X_{m,k-1}^{\prime\prime}}\otimes
\rho_{X_{m,k}X_{m,k}^{\prime\prime}BE}\otimes\rho_{X_{m,k+1}X_{m,k+1}%
^{\prime\prime}}\otimes\cdots\otimes\rho_{X_{M,K}X_{M,K}^{\prime\prime}},
\end{equation}
and the state of Bob's systems is%
\begin{equation}
\rho_{X^{MK}B}^{m,k}=\rho_{X_{1,1}}\otimes\cdots\otimes\rho_{X_{m,k-1}}%
\otimes\rho_{X_{m,k}B}\otimes\rho_{X_{m,k+1}}\otimes\cdots\otimes\rho
_{X_{M,K}}.
\end{equation}
For Bob to decode, he uses the position-based decoder to decode both the
message $m$ and the local key $k$. Let $\{\Lambda_{X^{M,K}B}^{m,k}\}_{m,k}$
denote his decoding POVM. By the reasoning from
Section~\ref{sec:random-assisted-public-classical}, as long as%
\begin{equation}
\log_{2}MK=I_{H}^{\varepsilon_{1}-\eta_{1}}(X;B)_{\rho}-\log_{2}%
(4\varepsilon_{1}/\eta_{1}^{2}), \label{eq:message-key-size}%
\end{equation}
where $\varepsilon_{1}\in(0,1)$ and $\eta_{1}\in(0,\varepsilon_{1})$, then we
have the following bound holding for all $m,k$:%
\begin{equation}
\operatorname{Tr}\{(I_{X^{M,K}B}-\Lambda_{X^{MK}B}^{m,k})\rho_{X^{MK}B}%
^{m,k}\}\leq\varepsilon_{1}, \label{eq:avg-exp-dec-err-private-code}%
\end{equation}
where $\Lambda_{X^{MK}B}^{m,k}$ is defined as in
Sections~\ref{sec:random-assisted-public-classical} and
\ref{sec:derandomize-public-classical-code}. By the reasoning from
Section~\ref{sec:derandomize-public-classical-code}, we can also write
\eqref{eq:avg-exp-dec-err-private-code} as%
\begin{equation}
\sum_{x_{1,1},\ldots,x_{M,K}}p_{X}(x_{1,1})\cdots p_{X}(x_{M,K})\left[
\frac{1}{MK}\sum_{m=1}^{M}\sum_{k=1}^{K}\operatorname{Tr}\{(I_{B}-\Omega
_{B}^{x_{m,k}})\rho_{B}^{x_{m,k}}\}\right]  \leq\varepsilon_{1},
\label{eq:avg-exp-dec-err-private-code-rewrite}%
\end{equation}
with $\Omega_{B}^{x_{m,k}}$ defined as in
Section~\ref{sec:derandomize-public-classical-code}. Define the following
measurement channels:%
\begin{align}
\mathcal{M}_{B\rightarrow\hat{M}}(\omega_{B})  &  \equiv\sum_{m,k}%
\operatorname{Tr}\{\Omega_{B}^{x_{m,k}}\omega_{B}\}|m\rangle\langle
m|_{\hat{M}},\\
\mathcal{M}_{B\rightarrow\hat{M}\hat{K}}^{\prime}(\omega_{B})  &  \equiv
\sum_{m,k}\operatorname{Tr}\{\Omega_{B}^{x_{m,k}}\omega_{B}\}|m\rangle\langle
m|_{\hat{M}}\otimes|k\rangle\langle k|_{\hat{K}},
\end{align}
with it being clear that $\operatorname{Tr}_{\hat{K}}\circ\mathcal{M}%
_{B\rightarrow\hat{M}\hat{K}}^{\prime}=\mathcal{M}_{B\rightarrow\hat{M}}$.
Consider that%
\begin{align}
&  \frac{1}{2}\left\Vert \mathcal{M}_{B\rightarrow\hat{M}\hat{K}}^{\prime
}\left(  \rho_{B}^{x_{m,k}}\right)  -|m\rangle\langle m|_{\hat{M}}%
\otimes|k\rangle\langle k|_{\hat{K}}\right\Vert _{1}\nonumber\\
&  =\frac{1}{2}\left\Vert \sum_{m^{\prime},k^{\prime}}\operatorname{Tr}%
\{\Omega_{B}^{x_{m^{\prime},k^{\prime}}}\rho_{B}^{x_{m,k}}\}|m^{\prime}%
\rangle\langle m^{\prime}|_{\hat{M}}\otimes|k^{\prime}\rangle\langle
k^{\prime}|_{\hat{K}}-|m\rangle\langle m|_{\hat{M}}\otimes|k\rangle\langle
k|_{\hat{K}}\right\Vert _{1}\\
&  =\frac{1}{2}\left\Vert \sum_{(m^{\prime},k^{\prime})\neq(m,k)}%
\operatorname{Tr}\{\Omega_{B}^{x_{m^{\prime},k^{\prime}}}\rho_{B}^{x_{m,k}%
}\}|m^{\prime}\rangle\langle m^{\prime}|_{\hat{M}}\otimes|k^{\prime}%
\rangle\langle k^{\prime}|_{\hat{K}}-(1-\operatorname{Tr}\{\Omega_{B}%
^{x_{m,k}}\rho_{B}^{x_{m,k}}\})|m\rangle\langle m|_{\hat{M}}\otimes
|k\rangle\langle k|_{\hat{K}}\right\Vert _{1}\\
&  =1-\operatorname{Tr}\{\Omega_{B}^{x_{m,k}}\rho_{B}^{x_{m,k}}\}\\
&  =\operatorname{Tr}\{(I_{B}-\Omega_{B}^{x_{m,k}})\rho_{B}^{x_{m,k}}\}.
\end{align}
Now averaging the above quantity over $m$, $k$, and $x_{1,1}$, \ldots,
$x_{M,K}$, and applying the condition in
\eqref{eq:avg-exp-dec-err-private-code-rewrite}, we get that%
\begin{equation}
\sum_{x_{1,1},\ldots,x_{M,K}}p_{X}(x_{1,1})\cdots p_{X}(x_{M,K})\left[
\frac{1}{MK}\sum_{m,k}\frac{1}{2}\left\Vert \mathcal{M}_{B\rightarrow\hat
{M}\hat{K}}^{\prime}\left(  \rho_{B}^{x_{m,k}}\right)  -|m\rangle\langle
m|_{\hat{M}}\otimes|k\rangle\langle k|_{\hat{K}}\right\Vert _{1}\right]
\leq\varepsilon_{1}. \label{eq:decod-priv-almost-there}%
\end{equation}
Applying convexity of the trace distance to bring the average over $k$ inside
and monotonicity with respect to partial trace over system $\hat{K}$\ to the
left-hand side of \eqref{eq:decod-priv-almost-there}, we find that%
\begin{multline}
\sum_{x_{1,1},\ldots,x_{M,K}}p_{X}(x_{1,1})\cdots p_{X}(x_{M,K})\left[
\frac{1}{M}\sum_{m=1}^{M}\frac{1}{2}\left\Vert (\operatorname{Tr}_{\hat{K}%
}\circ\mathcal{M}_{B\rightarrow\hat{M}\hat{K}}^{\prime})\left(  \frac{1}%
{K}\sum_{k=1}^{K}\rho_{B}^{x_{m,k}}\right)  -|m\rangle\langle m|_{\hat{M}%
}\right\Vert _{1}\right] \\
=\sum_{x_{1,1},\ldots,x_{M,K}}p_{X}(x_{1,1})\cdots p_{X}(x_{M,K})\left[
\frac{1}{M}\sum_{m=1}^{M}\frac{1}{2}\left\Vert \mathcal{M}_{B\rightarrow
\hat{M}}\left(  \frac{1}{K}\sum_{k=1}^{K}\rho_{B}^{x_{m,k}}\right)
-|m\rangle\langle m|_{\hat{M}}\right\Vert _{1}\right]  \leq\varepsilon_{1}.
\label{eq:rel-bound-priv-comm}%
\end{multline}
Let us define the state%
\begin{align}
\omega_{E}^{x_{m^{\prime}},x_{m}}  &  \equiv\frac{\frac{1}{K}\sum
_{k,k^{\prime}=1}^{K}\operatorname{Tr}_{B}\{\Omega_{B}^{x_{m^{\prime
},k^{\prime}}}\rho_{BE}^{x_{m,k}}\}}{q(x_{m^{\prime}}|x_{m})},\\
q(x_{m^{\prime}}|x_{m})  &  \equiv\frac{1}{K}\sum_{k,k^{\prime}=1}%
^{K}\operatorname{Tr}\{\Omega_{B}^{x_{m^{\prime},k^{\prime}}}\rho
_{BE}^{x_{m,k}}\}.
\end{align}
Consider that%
\begin{equation}
\sum_{m^{\prime}}q(x_{m^{\prime}}|x_{m})\omega_{E}^{x_{m^{\prime}},x_{m}%
}=\frac{1}{K}\sum_{k=1}^{K}\rho_{E}^{x_{m,k}}.
\end{equation}
Then we can write%
\begin{equation}
\mathcal{M}_{B\rightarrow\hat{M}}\left(  \frac{1}{K}\sum_{k=1}^{K}\rho
_{BE}^{x_{m,k}}\right)  =\sum_{m^{\prime}}q(x_{m^{\prime}}|x_{m})|m^{\prime
}\rangle\langle m^{\prime}|_{\hat{M}}\otimes\omega_{E}^{x_{m^{\prime}},x_{m}},
\end{equation}
so that%
\begin{equation}
\mathcal{M}_{B\rightarrow\hat{M}}\left(  \frac{1}{K}\sum_{k=1}^{K}\rho
_{B}^{x_{m,k}}\right)  =\sum_{m^{\prime}}q(x_{m^{\prime}}|x_{m})|m^{\prime
}\rangle\langle m^{\prime}|_{\hat{M}}.
\end{equation}
Using these observations, we can finally write%
\begin{align}
&  \frac{1}{M}\sum_{m=1}^{M}\frac{1}{2}\left\Vert \mathcal{M}_{B\rightarrow
\hat{M}}\left(  \frac{1}{K}\sum_{k=1}^{K}\rho_{BE}^{x_{m,k}}\right)
-|m\rangle\langle m|_{\hat{M}}\otimes\frac{1}{K}\sum_{k=1}^{K}\rho
_{E}^{x_{m,k}}\right\Vert _{1}\nonumber\\
&  =\frac{1}{M}\sum_{m=1}^{M}\frac{1}{2}\left\Vert \sum_{m^{\prime}%
}q(x_{m^{\prime}}|x_{m})|m^{\prime}\rangle\langle m^{\prime}|_{\hat{M}}%
\otimes\omega_{E}^{x_{m^{\prime}},x_{m}}-|m\rangle\langle m|_{\hat{M}}%
\otimes\sum_{m^{\prime}}q(x_{m^{\prime}}|x_{m})\omega_{E}^{x_{m^{\prime}%
},x_{m}}\right\Vert _{1}\\
&  \leq\frac{1}{M}\sum_{m=1}^{M}\sum_{m^{\prime}}q(x_{m^{\prime}}%
|x_{m})\left[  \frac{1}{2}\left\Vert |m^{\prime}\rangle\langle m^{\prime
}|_{\hat{M}}\otimes\omega_{E}^{x_{m^{\prime}},x_{m}}-|m\rangle\langle
m|_{\hat{M}}\otimes\omega_{E}^{x_{m^{\prime}},x_{m}}\right\Vert _{1}\right] \\
&  =\frac{1}{M}\sum_{m=1}^{M}\sum_{m^{\prime}}q(x_{m^{\prime}}|x_{m})\left[
\frac{1}{2}\left\Vert |m^{\prime}\rangle\langle m^{\prime}|_{\hat{M}%
}-|m\rangle\langle m|_{\hat{M}}\right\Vert _{1}\right] \\
&  =\frac{1}{M}\sum_{m=1}^{M}\sum_{m^{\prime}\neq m}q(x_{m^{\prime}}|x_{m})\\
&  =\frac{1}{M}\sum_{m=1}^{M}\frac{1}{2}\left\Vert \mathcal{M}_{B\rightarrow
\hat{M}}\left(  \frac{1}{K}\sum_{k=1}^{K}\rho_{B}^{x_{m,k}}\right)
-|m\rangle\langle m|_{\hat{M}}\right\Vert _{1}.
\end{align}
Combining with \eqref{eq:rel-bound-priv-comm}, the above development implies
that%
\begin{multline}
\sum_{x_{1,1},\ldots,x_{M,K}}p_{X}(x_{1,1})\cdots p_{X}(x_{M,K}%
)\label{eq:final-rel-decode-condition}\\
\times\left[  \frac{1}{M}\sum_{m=1}^{M}\frac{1}{2}\left\Vert \mathcal{M}%
_{B\rightarrow\hat{M}}\left(  \frac{1}{K}\sum_{k=1}^{K}\rho_{BE}^{x_{m,k}%
}\right)  -|m\rangle\langle m|_{\hat{M}}\otimes\frac{1}{K}\sum_{k=1}^{K}%
\rho_{E}^{x_{m,k}}\right\Vert _{1}\right]  \leq\varepsilon_{1}.
\end{multline}

Now we consider the state on Eve's systems and the analysis of privacy. If $m$
and $k$ are fixed, then her state is%
\begin{equation}
\rho_{X^{MK}E}^{m,k}=\rho_{X_{1,1}}\otimes\cdots\otimes\rho_{X_{m,k-1}}%
\otimes\rho_{X_{m,k}E}\otimes\rho_{X_{m,k+1}}\otimes\cdots\otimes\rho
_{X_{M,K}}.
\end{equation}
(For simplicity of notation, in the above and what follows we are labeling her
systems $X^{\prime\prime}$ as $X$.) However, $k$ is chosen uniformly at
random, and so conditioned on the message $m$ being fixed, the state of Eve's
systems is as follows:%
\begin{align}
\rho_{X^{MK}E}^{m}  &  \equiv\frac{1}{K}\sum_{k=1}^{K}\rho_{X^{MK}E}^{m,k}\\
&  =\rho_{X_{1,1}}\otimes\cdots\otimes\rho_{X_{m-1,K}}\nonumber\\
&  \qquad\otimes\left[  \frac{1}{K}\sum_{k=1}^{K}\rho_{X_{m,1}}\otimes
\cdots\otimes\rho_{X_{m,k-1}}\otimes\rho_{X_{m,k}E}\otimes\rho_{X_{m,k+1}%
}\otimes\cdots\otimes\rho_{X_{m,K}}\right] \nonumber\\
&  \qquad\otimes\rho_{X_{m+1,1}}\otimes\cdots\otimes\rho_{X_{M,K}}.
\end{align}
We would like to show for $\varepsilon_{2}\in(0,1)$ that%
\begin{equation}
\frac{1}{2}\left\Vert \rho_{X^{MK}E}^{m}-\rho_{X^{MK}}\otimes\widetilde{\rho
}_{E}\right\Vert _{1}\leq\varepsilon_{2},
\end{equation}
for some state $\widetilde{\rho}_{E}$. By the invariance of the trace distance
with respect to tensor-product states, i.e.,%
\begin{equation}
\left\Vert \sigma\otimes\tau-\omega\otimes\tau\right\Vert _{1}=\left\Vert
\sigma-\omega\right\Vert _{1}, \label{eq:trace-dist-prop}%
\end{equation}
we find that%
\begin{align}
&  \frac{1}{2}\left\Vert \rho_{X^{MK}E}^{m}-\rho_{X^{MK}}\otimes
\widetilde{\rho}_{E}\right\Vert _{1}\\
&  =\frac{1}{2}\left\Vert \rho_{X_{m,1}\cdots X_{m,K}E}^{m}-\rho
_{X_{m,1}\cdots X_{m,K}}\otimes\widetilde{\rho}_{E}\right\Vert _{1}\\
&  =\frac{1}{2}\left\Vert \frac{1}{K}\sum_{k=1}^{K}\rho_{X_{m,1}}\otimes
\cdots\otimes\rho_{X_{m,k-1}}\otimes\left(  \rho_{X_{m,k}E}-\rho_{X_{m,k}%
}\otimes\widetilde{\rho}_{E}\right)  \otimes\rho_{X_{m,k+1}}\otimes
\cdots\otimes\rho_{X_{m,K}}\right\Vert _{1}.
\end{align}
From Lemma~\ref{thm:convex-split} and the relation in
\eqref{eq:TD-to-PD}\ between trace distance and purified distance, we find
that if we pick $K$ such that%
\begin{equation}
\log_{2}K=\widetilde{I}_{\max}^{\sqrt{\varepsilon_{2}}-\eta_{2}}(E;X)_{\rho
}+2\log_{2}( 1/\eta_{2}) , \label{eq:key-size}%
\end{equation}
then we are guaranteed that%
\begin{equation}
\frac{1}{2}\left\Vert \rho_{X^{MK}E}^{m}-\rho_{X^{MK}}\otimes\widetilde{\rho
}_{E}\right\Vert _{1}\leq\sqrt{\varepsilon_{2}},
\end{equation}
where $\widetilde{\rho}_{E}$ is some state such that $P(\widetilde{\rho}%
_{E},\rho_{E})\leq\sqrt{\varepsilon_{2}}-\eta_{2}$.

Consider that we can rewrite%
\begin{align}
&  \frac{1}{2}\left\Vert \rho_{X^{MK}E}^{m}-\rho_{X^{MK}}\otimes
\widetilde{\rho}_{E}\right\Vert _{1}\\
&  =\frac{1}{2}\left\Vert \frac{1}{K}\sum_{k=1}^{K}\sum_{x_{1,1}\cdots
x_{M,K}}p_{X}(x_{1,1})\cdots p_{X}(x_{M,K})|x_{1,1}\cdots x_{M,K}%
\rangle\langle x_{1,1}\cdots x_{M,K}|_{X^{M,K}}\otimes\left(  \rho
_{E}^{x_{m,k}}-\widetilde{\rho}_{E}\right)  \right\Vert _{1}\\
&  =\frac{1}{2}\left\Vert \sum_{x_{1,1}\cdots x_{M,K}}p_{X}(x_{1,1})\cdots
p_{X}(x_{M,K})|x_{1,1}\cdots x_{M,K}\rangle\langle x_{1,1}\cdots
x_{M,K}|_{X^{M,K}}\otimes\left(  \frac{1}{K}\sum_{k=1}^{K}\rho_{E}^{x_{m,k}%
}-\widetilde{\rho}_{E}\right)  \right\Vert _{1}\\
&  =\sum_{x_{1,1}\cdots x_{M,K}}p_{X}(x_{1,1})\cdots p_{X}(x_{M,K})\left[
\frac{1}{2}\left\Vert \frac{1}{K}\sum_{k=1}^{K}\rho_{E}^{x_{m,k}}%
-\widetilde{\rho}_{E}\right\Vert _{1}\right]  \leq\sqrt{\varepsilon_{2}}.
\label{eq:privacy-condition-proof}%
\end{align}
Applying \eqref{eq:trace-dist-prop} to \eqref{eq:privacy-condition-proof}, we
find that%
\begin{equation}
\sum_{x_{1,1}\cdots x_{M,K}}p_{X}(x_{1,1})\cdots p_{X}(x_{M,K})\left[
\frac{1}{2}\left\Vert |m\rangle\langle m|_{\hat{M}}\otimes\frac{1}{K}%
\sum_{k=1}^{K}\rho_{E}^{x_{m,k}}-|m\rangle\langle m|_{\hat{M}}\otimes
\widetilde{\rho}_{E}\right\Vert _{1}\right]  \leq\sqrt{\varepsilon_{2}}.
\label{eq:final-privacy-condition}%
\end{equation}
Putting together \eqref{eq:message-key-size}, \eqref{eq:key-size},
\eqref{eq:final-rel-decode-condition}, and \eqref{eq:final-privacy-condition},
we find that if
\begin{align}
\log_{2}M  &  =I_{H}^{\varepsilon_{1}-\eta_{1}}(X;B)_{\rho}-\log
_{2}(4\varepsilon_{1}/\eta_{1}^{2})-\left[  \widetilde{I}_{\max}%
^{\sqrt{\varepsilon_{2}}-\eta_{2}}(E;X)_{\rho}+2\log_{2}( 1/\eta_{2}) \right]
\\
&  =I_{H}^{\varepsilon_{1}-\eta_{1}}(X;B)_{\rho}-\widetilde{I}_{\max}%
^{\sqrt{\varepsilon_{2}}-\eta_{2}}(E;X)_{\rho}-\log_{2}(4\varepsilon_{1}%
/\eta_{1}^{2})-2\log_{2}( 1/\eta_{2}) ,
\end{align}
then we have by the triangle inequality that%
\begin{equation}
\sum_{x_{1,1},\ldots,x_{M,K}}p_{X}(x_{1,1})\cdots p_{X}(x_{M,K})\left[
\frac{1}{2}\left\Vert \mathcal{M}_{B\rightarrow\hat{M}}\left(  \frac{1}{K}%
\sum_{k=1}^{K}\rho_{BE}^{x_{m,k}}\right)  -|m\rangle\langle m|_{\hat{M}%
}\otimes\widetilde{\rho}_{E}\right\Vert _{1}\right]  \leq\varepsilon_{1}%
+\sqrt{\varepsilon_{2}}.
\end{equation}
So this gives what is achievable with shared randomness (again, no difference
between average and maximal error if shared randomness is allowed).

We now show how to derandomize the code. We take the above and average over
all messages $m$. We find that%
\begin{align}
&  \varepsilon_{1}+\sqrt{\varepsilon_{2}}\nonumber\\
&  \geq\frac{1}{M}\sum_{m=1}^{M}\sum_{x_{1,1},\ldots,x_{M,K}}p_{X}%
(x_{1,1})\cdots p_{X}(x_{M,K})\left[  \frac{1}{2}\left\Vert \mathcal{M}%
_{B\rightarrow\hat{M}}\left(  \frac{1}{K}\sum_{k=1}^{K}\rho_{BE}^{x_{m,k}%
}\right)  -|m\rangle\langle m|_{\hat{M}}\otimes\widetilde{\rho}_{E}\right\Vert
_{1}\right] \\
&  =\sum_{x_{1,1},\ldots,x_{M,K}}p_{X}(x_{1,1})\cdots p_{X}(x_{M,K})\left(
\frac{1}{M}\sum_{m=1}^{M}\left[  \frac{1}{2}\left\Vert \mathcal{M}%
_{B\rightarrow\hat{M}}\left(  \frac{1}{K}\sum_{k=1}^{K}\rho_{BE}^{x_{m,k}%
}\right)  -|m\rangle\langle m|_{\hat{M}}\otimes\widetilde{\rho}_{E}\right\Vert
_{1}\right]  \right)  .
\end{align}
So we can conclude that there exist particular values $x_{1,1}$, \ldots,
$x_{M,K}$ such that%
\begin{equation}
\frac{1}{M}\sum_{m=1}^{M}\left[  \frac{1}{2}\left\Vert \mathcal{M}%
_{B\rightarrow\hat{M}}\left(  \frac{1}{K}\sum_{k=1}^{K}\rho_{BE}^{x_{m,k}%
}\right)  -|m\rangle\langle m|_{\hat{M}}\otimes\widetilde{\rho}_{E}\right\Vert
_{1}\right]  \leq\varepsilon_{1}+\sqrt{\varepsilon_{2}}.
\label{eq:final-wiretap-perf}%
\end{equation}
Thus, our final conclusion is that the number of achievable bits that can be
sent such that the privacy error is no larger than $\varepsilon_{1}%
+\sqrt{\varepsilon_{2}}$ is equal to%
\begin{equation}
\log_{2}M=I_{H}^{\varepsilon_{1}-\eta_{1}}(X;B)_{\rho}-\widetilde{I}_{\max
}^{\sqrt{\varepsilon_{2}}-\eta_{2}}(E;X)_{\rho}-\log_{2}(4\varepsilon_{1}%
/\eta_{1}^{2})-2\log_{2}( 1/\eta_{2}) . \label{eq:rate-bound-wiretap}%
\end{equation}

\subsection{Second-order asymptotics for private classical communication}

In this section, I show how the lower bound on one-shot private capacity
leads to a non-trivial lower bound on the second-order coding rate of private
communication over an i.i.d.~cq wiretap channel. I also show how the bounds simplify for pure-state cq wiretap channels and when using binary phase-shift keying as a coding strategy for private communication over a pure-loss bosonic channel.

Applying
Lemma~\ref{lem:alt-dmax-to-dmax-smooth} to \eqref{eq:rate-bound-wiretap} with
$\gamma\in(0,\sqrt{\varepsilon}-\eta)$, we can take%
\begin{align}
\log_{2}M  &  =I_{H}^{\varepsilon_{1}-\eta_{1}}(X;B)_{\rho}-\widetilde
{I}_{\max}^{\sqrt{\varepsilon_{2}}-\eta_{2}}(E;X)_{\rho}-\log_{2}%
(4\varepsilon_{1}/\eta_{1}^{2})-2\log_{2}( 1/\eta_{2})\\
&  \geq I_{H}^{\varepsilon_{1}-\eta_{1}}(X;B)_{\rho}-D_{\max}^{\sqrt
{\varepsilon_{2} }-\eta_{2}-\gamma}(\rho_{XE}\Vert\rho_{X}\otimes\rho
_{E})-\log_{2}(4\varepsilon_{1}/\eta_{1}^{2})-2\log_{2}( 1/\eta_{2}) -\log
_{2}( 3/\gamma^{2}) .
\end{align}
while still achieving the performance in \eqref{eq:final-wiretap-perf}.

Substituting an i.i.d.~cq wiretap channel into the one-shot bounds, evaluating
for such a case and using the expansions for $I_{H}^{\varepsilon}$ in
\eqref{eq:MI-expand} and $D_{\max}^{\varepsilon}$ in \eqref{eq:MI-expand-2},
while taking $\eta_{1}=\eta_{2}=\gamma=1/\sqrt{n}$, for sufficiently large
$n$, we get that%
\begin{multline}
\log_{2}M_{\operatorname{priv}}^{\ast}(n,\varepsilon_{1}+\sqrt{\varepsilon
_{2}})\geq n\left[  I(X;B)_{\rho}-I(X;E)_{\rho}\right]
\label{eq:second-order-bound}\\
+\sqrt{nV(X;B)_{\rho}}\Phi^{-1}(\varepsilon_{1})+\sqrt{nV(X;E)_{\rho}}%
\Phi^{-1}(\varepsilon_{2})+O(\log n).
\end{multline}

\subsubsection{Example:\ Pure-state cq wiretap channel}

Let us consider applying the inequality in \eqref{eq:second-order-bound} to a
cq pure-state wiretap channel of the following form:%
\begin{equation}
x\rightarrow|\psi^{x}\rangle\langle\psi^{x}|_{B}\otimes|\varphi^{x}%
\rangle\langle\varphi^{x}|_{E},\label{eq:pure-state-cq-channel}%
\end{equation}
in which the classical input $x$ leads to a pure quantum state $|\psi
^{x}\rangle\langle\psi^{x}|_{B}$ for Bob and a pure quantum state
$|\varphi^{x}\rangle\langle\varphi^{x}|_{E}$ for Eve. This channel may seem a
bit particular, but we discuss in the next section how one can induce such a
channel from a practically relevant channel, known as the pure-loss bosonic
channel. In order to apply the inequality in \eqref{eq:second-order-bound} to
the channel in \eqref{eq:pure-state-cq-channel}, we fix a distribution
$p_{X}(x)$ over the input symbols, leading to the following classical--quantum
state:%
\begin{equation}
\rho_{XBE}\equiv\sum_{x}p_{X}(x)|x\rangle\langle x|_{X}\otimes|\psi^{x}%
\rangle\langle\psi^{x}|_{B}\otimes|\varphi^{x}\rangle\langle\varphi^{x}|_{E}.
\end{equation}
It is well known and straightforward to calculate that the following
simplifications occur%
\begin{align}
I(X;B)_{\rho}  & =H(B)_{\rho}=H(\rho_{B}),\\
I(X;E)_{\rho}  & =H(E)_{\rho}=H(\rho_{E}),
\end{align}
where $H(\sigma) \equiv - \operatorname{Tr} \{ \sigma \log_2 \sigma\}$ denotes the quantum entropy of a state $\sigma$ and %
\begin{align}
\rho_{B}  & =\sum_{x}p_{X}(x)|\psi^{x}\rangle\langle\psi^{x}|_{B},\\
\rho_{E}  & =\sum_{x}p_{X}(x)|\varphi^{x}\rangle\langle\varphi^{x}|_{E}.
\end{align}
Proposition~\ref{prop:variance-simplify-pure-state-cq} below demonstrates that a
similar simplification occurs for the information variance quantities in
\eqref{eq:second-order-bound}, in the special case of a pure-state cq wiretap
channel. By employing it, we find the following lower bound on the
second-order coding rate for a pure-state cq wiretap channel:%
\begin{multline}
\log_{2}M_{\operatorname{priv}}^{\ast}(n,\varepsilon_{1}+\sqrt{\varepsilon
_{2}})\geq n\left[  H(\rho_{B})-H(\rho_{E})\right]
\label{eq:pure-state-cq-2nd-order-bnd}\\
+\sqrt{nV(\rho_{B})}\Phi^{-1}(\varepsilon_{1})+\sqrt{nV(\rho_{E})}\Phi
^{-1}(\varepsilon_{2})+O(\log n),
\end{multline}
where $V(\rho_{B})$ and $V(\rho_{E})$ are defined from
\eqref{eq:entropy-var-def} below.

\begin{proposition}
\label{prop:variance-simplify-pure-state-cq}Let%
\begin{equation}
\rho_{XB}=\sum_{x}p_{X}(x)|x\rangle\langle x|_{X}\otimes|\psi^{x}%
\rangle\langle\psi^{x}|_{B}\label{eq:special-cq-state}%
\end{equation}
be a classical--quantum state corresponding to a pure-state ensemble
$\{p_{X}(x),|\psi^{x}\rangle_{B}\}_{x}$. Then the Holevo information variance
$V(X;B)_{\rho}=V(\rho_{XB}\Vert\rho_{X}\otimes\rho_{B})$\ is equal to the
entropy variance $V(\rho_{B})$ of the expected state $\rho_{B}=\sum_{x}%
p_{X}(x)|\psi^{x}\rangle\langle\psi^{x}|_{B}$, where%
\begin{equation}
V(\sigma)=\operatorname{Tr}\{\sigma\left[  -\log_{2}\sigma-H(\sigma)\right]  ^{2}\}.\label{eq:entropy-var-def}%
\end{equation}
That is, when $\rho_{XB}$ takes the special form in
\eqref{eq:special-cq-state}, the following equality holds%
\begin{equation}
V(X;B)_{\rho}=V(\rho_{B}).\label{eq:prop-statement-pure-state-cq}%
\end{equation}

\end{proposition}

\begin{proof}
For the cq state in \eqref{eq:special-cq-state}, consider that $I(X;B)_{\rho
}=H(B)_{\rho}=H(\rho_{B})$. Furthermore, we have that%
\begin{align}
\log_{2}\rho_{XB}  & =\log_{2}\left[  \sum_{x}p_{X}(x)|x\rangle\langle
x|_{X}\otimes|\psi^{x}\rangle\langle\psi^{x}|_{B}\right]  \\
& =\sum_{x}\log_{2}\left(  p_{X}(x)\right)  |x\rangle\langle x|_{X}%
\otimes|\psi^{x}\rangle\langle\psi^{x}|_{B},
\end{align}
which holds because the eigenvectors of $\rho_{XB}$ are $\left\{
|x\rangle_{X}\otimes|\psi^{x}\rangle_{B}\right\}  _{x}$. Then%
\begin{align}
V(X;B)  & =V(\rho_{XB}\Vert\rho_{X}\otimes\rho_{B})\\
& =\operatorname{Tr}\{\rho_{XB}\left[  \log_{2}\rho_{XB}-\log_{2}\left(
\rho_{X}\otimes\rho_{B}\right)  \right]  ^{2}\}-\left[  I(X;B)_{\rho}\right]
^{2}\\
& =\operatorname{Tr}\{\rho_{XB}\left[  \log_{2}\rho_{XB}-\log_{2}\rho
_{X}\otimes I_{B}-I_{X}\otimes\log_{2}\rho_{B}\right]  ^{2}\}-\left[
H(B)_{\rho}\right]  ^{2}.\label{eq:1st-step-ent-var}%
\end{align}
By direct calculation, we have that%
\begin{align}
& \log_{2}\rho_{XB}-\log_{2}\rho_{X}\otimes I_{B}-I_{X}\otimes\log_{2}\rho
_{B}\nonumber\\
& =\sum_{x}\log_{2}\left(  p_{X}(x)\right)  |x\rangle\langle x|_{X}%
\otimes|\psi^{x}\rangle\langle\psi^{x}|_{B}-\sum_{x}\log_{2}\left[
p_{X}(x)\right]  |x\rangle\langle x|_{X}\otimes I_{B}-\sum_{x}|x\rangle\langle
x|_{X}\otimes\log_{2}\rho_{B}\\
& =-\sum_{x}|x\rangle\langle x|_{X}\otimes\left[  \log_{2}\left(
p_{X}(x)\right)  \left(  I_{B}-|\psi^{x}\rangle\langle\psi^{x}|_{B}\right)
+\log_{2}\rho_{B}\right]  .
\end{align}
Observe that $I_{B}-|\psi^{x}\rangle\langle\psi^{x}|_{B}$ is the projection
onto the space orthogonal to $|\psi^{x}\rangle_{B}$. Then we find that%
\begin{align}
& \left[  \log_{2}\rho_{XB}-\log_{2}\rho_{X}\otimes I_{B}-I_{X}\otimes\log
_{2}\rho_{B}\right]  ^{2}\nonumber\\
& =\left[  -\sum_{x}|x\rangle\langle x|_{X}\otimes\left[  \log_{2}\left(
p_{X}(x)\right)  \left(  I_{B}-|\psi^{x}\rangle\langle\psi^{x}|_{B}\right)
+\log_{2}\rho_{B}\right]  \right]  ^{2}\\
& =\sum_{x}|x\rangle\langle x|_{X}\otimes\left[  \log_{2}\left(
p_{X}(x)\right)  \left(  I_{B}-|\psi^{x}\rangle\langle\psi^{x}|_{B}\right)
+\log_{2}\rho_{B}\right]  ^{2}.
\end{align}
Furthermore, we have that%
\begin{multline}
\left[  \log_{2}\left(  p_{X}(x)\right)  \left(  I_{B}-|\psi^{x}\rangle
\langle\psi^{x}|_{B}\right)  +\log_{2}\rho_{B}\right]  ^{2}%
\label{eq:entropy-var-expand}\\
=\left[  \log_{2}\left(  p_{X}(x)\right)  \right]  ^{2}\left(  I_{B}-|\psi
^{x}\rangle\langle\psi^{x}|_{B}\right)  +\log_{2}\left(  p_{X}(x)\right)
\left(  I_{B}-|\psi^{x}\rangle\langle\psi^{x}|_{B}\right)  \left(  \log
_{2}\rho_{B}\right)  \\
+\log_{2}\left(  p_{X}(x)\right)  \left(  \log_{2}\rho_{B}\right)  \left(
I_{B}-|\psi^{x}\rangle\langle\psi^{x}|_{B}\right)  +\left[  \log_{2}\rho
_{B}\right]  ^{2}%
\end{multline}
So then, by direct calculation,%
\begin{align}
& \operatorname{Tr}\{\rho_{XB}\left[  \log_{2}\rho_{XB}-\log_{2}\rho
_{X}\otimes I_{B}-I_{X}\otimes\log_{2}\rho_{B}\right]  ^{2}\}\\
& =\operatorname{Tr}\left\{  \left[  \sum_{x^{\prime}}p_{X}(x^{\prime
})|x^{\prime}\rangle\langle x^{\prime}|_{X}\otimes|\psi^{x^{\prime}}%
\rangle\langle\psi^{x^{\prime}}|_{B}\right]  \left[  \sum_{x}|x\rangle\langle
x|_{X}\otimes\left[  \log_{2}\left(  p_{X}(x)\right)  \left(  I_{B}-|\psi
^{x}\rangle\langle\psi^{x}|_{B}\right)  +\log_{2}\rho_{B}\right]  ^{2}\right]
\right\}  \\
& =\sum_{x}p_{X}(x)\operatorname{Tr}\left\{  |\psi^{x}\rangle\langle\psi
^{x}|_{B}\left[  \left[  \log_{2}\left(  p_{X}(x)\right)  \left(  I_{B}%
-|\psi^{x}\rangle\langle\psi^{x}|_{B}\right)  +\log_{2}\rho_{B}\right]
^{2}\right]  \right\}  \\
& =\sum_{x}p_{X}(x)\operatorname{Tr}\left\{  |\psi^{x}\rangle\langle\psi
^{x}|_{B}\left[  \log_{2}\rho_{B}\right]  ^{2}\right\}  \\
& =\operatorname{Tr}\{\rho_{B}\left[  \log_{2}\rho_{B}\right]  ^{2}%
\}.\label{eq:last-step-ent-var}%
\end{align}
In the second-to-last equality, we used the expansion in
\eqref{eq:entropy-var-expand} and the fact that $|\psi^{x}\rangle\langle
\psi^{x}|_{B}$ and $I_{B}-|\psi^{x}\rangle\langle\psi^{x}|_{B}$ are
orthogonal. Finally, putting together \eqref{eq:1st-step-ent-var}\ and
\eqref{eq:last-step-ent-var}, we conclude \eqref{eq:prop-statement-pure-state-cq}.
\end{proof}

\subsubsection{Example: Pure-loss bosonic channel}

We can induce a pure-state cq wiretap channel from a pure-loss bosonic
channel. In what follows, we consider a coding scheme called binary phase-shift keying
(BPSK). Let us recall just the basic facts needed from Gaussian quantum
information to support the argument that follows (a curious reader can consult
\cite{S17}\ for further details). The pure-loss channel of transmissivity
$\eta\in\left(  0,1\right)  $ is such that if the sender inputs a coherent
state $|\alpha\rangle$ with $\alpha\in\mathbb{C}$, then the outputs for Bob
and Eve are the coherent states $|\sqrt{\eta}\alpha\rangle_{B}$ and
$|\sqrt{1-\eta}\alpha\rangle_{E}$, respectively. Note that the overlap of any
two coherent states $|\alpha\rangle$ and $|\beta\rangle$ is equal to
$\left\vert \left\langle \alpha|\beta\right\rangle \right\vert ^{2}%
=e^{-\left\vert \alpha-\beta\right\vert ^{2}}$, and this is in fact the main
quantity that we need to evaluate the information quantities in
\eqref{eq:pure-state-cq-2nd-order-bnd}. The average photon number of a
coherent state $|\alpha\rangle$ is equal to $\left\vert \alpha\right\vert
^{2}$. A BPSK-coding scheme induces the following pure-state cq wiretap
channel from the pure-loss channel:%
\begin{align}
0  & \rightarrow|\alpha\rangle_{A}\rightarrow|\sqrt{\eta}\alpha\rangle
_{B}\otimes|\sqrt{1-\eta}\alpha\rangle_{E},\\
1  & \rightarrow|-\alpha\rangle_{A}\rightarrow|-\sqrt{\eta}\alpha\rangle
_{B}\otimes|-\sqrt{1-\eta}\alpha\rangle_{E}.
\end{align}
That is, if the sender would like to transmit the symbol \textquotedblleft%
0,\textquotedblright\ then she prepares the coherent state $|\alpha\rangle
_{A}$ at the input, and the physical channel prepares the coherent state
$|\sqrt{\eta}\alpha\rangle_{B}$ for Bob and $|\sqrt{1-\eta}\alpha\rangle_{E}$
for Eve. A similar explanation holds for when the sender inputs the symbol
\textquotedblleft1.\textquotedblright\ A BPSK-coding scheme is such that
the distribution $p_{X}(x)$ is unbiased:\ there is an equal probability 1/2 to
pick \textquotedblleft0\textquotedblright\ or \textquotedblleft%
1\textquotedblright\ when selecting codewords. Thus, the expected density
operators at the output for Bob and Eve are respectively as follows:%
\begin{align}
\rho_{B}  & =\frac{1}{2}\left(  |\sqrt{\eta}\alpha\rangle\langle\sqrt{\eta
}\alpha|_{B}+|-\sqrt{\eta}\alpha\rangle\langle-\sqrt{\eta}\alpha|_{B}\right)
,\\
\rho_{E}  & =\frac{1}{2}\left(  |\sqrt{1-\eta}\alpha\rangle\langle\sqrt
{1-\eta}\alpha|_{B}+|-\sqrt{1-\eta}\alpha\rangle\langle-\sqrt{1-\eta}%
\alpha|_{B}\right)  .
\end{align}
A straightforward computation reveals that the eigenvalues for $\rho_{B}$ are
a function only of the overlap $\left\vert \langle-\sqrt{\eta}\alpha
|\sqrt{\eta}\alpha\rangle\right\vert ^{2}=e^{-4\eta\left\vert \alpha
\right\vert ^{2}}\equiv e^{-4\eta\bar{n}}$ and are equal to \cite{GW12}
\begin{equation}
p^{B}(\eta,\bar{n})\equiv\frac{1}{2}\left(  1+e^{-2\eta\bar{n}}\right)
,\qquad 1-p^{B}(\eta,\bar{n})
=
\frac{1}{2}\left(  1-e^{-2\eta\bar{n}}\right)  .
\end{equation}
Similarly, the eigenvalues of $\rho_{E}$ are given by%
\begin{equation}
p^{E}(\eta,\bar{n})\equiv\frac{1}{2}\left(  1+e^{-2\left(  1-\eta\right)
\bar{n}}\right)  ,\qquad
1-p^{E}(\eta,\bar{n})
=
\frac{1}{2}\left(  1-e^{-2\left(  1-\eta\right)
\bar{n}}\right)  .
\end{equation}
We can then immediately plug in to \eqref{eq:pure-state-cq-2nd-order-bnd} to
find a lower bound on the second-order coding rate for private communication
over the pure-loss bosonic channel:%
\begin{multline}
\log_{2}M_{\operatorname{priv}}^{\ast}(n,\varepsilon_{1}+\sqrt{\varepsilon
_{2}})\geq n\left[  h_{2}(p^{B}(\eta,\bar{n}))-h_{2}(p^{E}(\eta,\bar
{n}))\right]  \\
+\sqrt{nv_{2}(p^{B}(\eta,\bar{n}))}\Phi^{-1}(\varepsilon_{1})+\sqrt
{nv_{2}(p^{E}(\eta,\bar{n}))}\Phi^{-1}(\varepsilon_{2})+O(\log n),\label{eq:pure-loss-2nd-order}
\end{multline}
where $h_{2}$ and $v_{2}$ respectively denote the binary entropy and binary
entropy variance:%
\begin{align}
h_{2}(\gamma)  & \equiv-\gamma\log_{2}\gamma-(1-\gamma)\log_{2}(1-\gamma),\\
v_{2}(\gamma)  & \equiv\gamma\left[  \log_{2}\gamma+h_{2}(\gamma)\right]
^{2}+(1-\gamma)\left[  \log_{2}\left(  1-\gamma\right)  +h_{2}(\gamma)\right]
^{2}.
\end{align}
A benchmark against which we can compare the performance of a BPSK\ code with
$\left\vert \alpha\right\vert ^{2}=\bar{n}$ is the energy-constrained private
capacity of a pure-loss bosonic channel \cite{wilde2016energy}, given by%
\begin{equation}
g(\eta\bar{n})-g((1-\eta)\bar{n}), \label{eq:actual-p-cap}
\end{equation}
where $g(x)\equiv(x+1)\log_{2}(x+1)-x\log_{2}x$. Figure~\ref{fig:results} plots the normal approximation \cite{polyanskiy10} of the lower bound on the second-order
coding rate of BPSK\ coding for various parameter choices for $\varepsilon
_{1}$, $\varepsilon_{2}$, $\eta$, and $\bar{n}$, comparing it against the
asymptotic performance of BPSK\ and the actual energy-constrained private capacity in \eqref{eq:actual-p-cap}. The normal approximation consists of all terms in \eqref{eq:pure-loss-2nd-order} besides the $O(\log n)$ term and typically serves as a good approximation for non-asymptotic capacity even for small values of $n$ (when \eqref{eq:pure-loss-2nd-order} is not necessarily valid), as previously observed in \cite{polyanskiy10,TH12,TBR15}.

\begin{figure}[ptb]
\begin{center}
\subfloat[]{\includegraphics[width=.47\columnwidth]{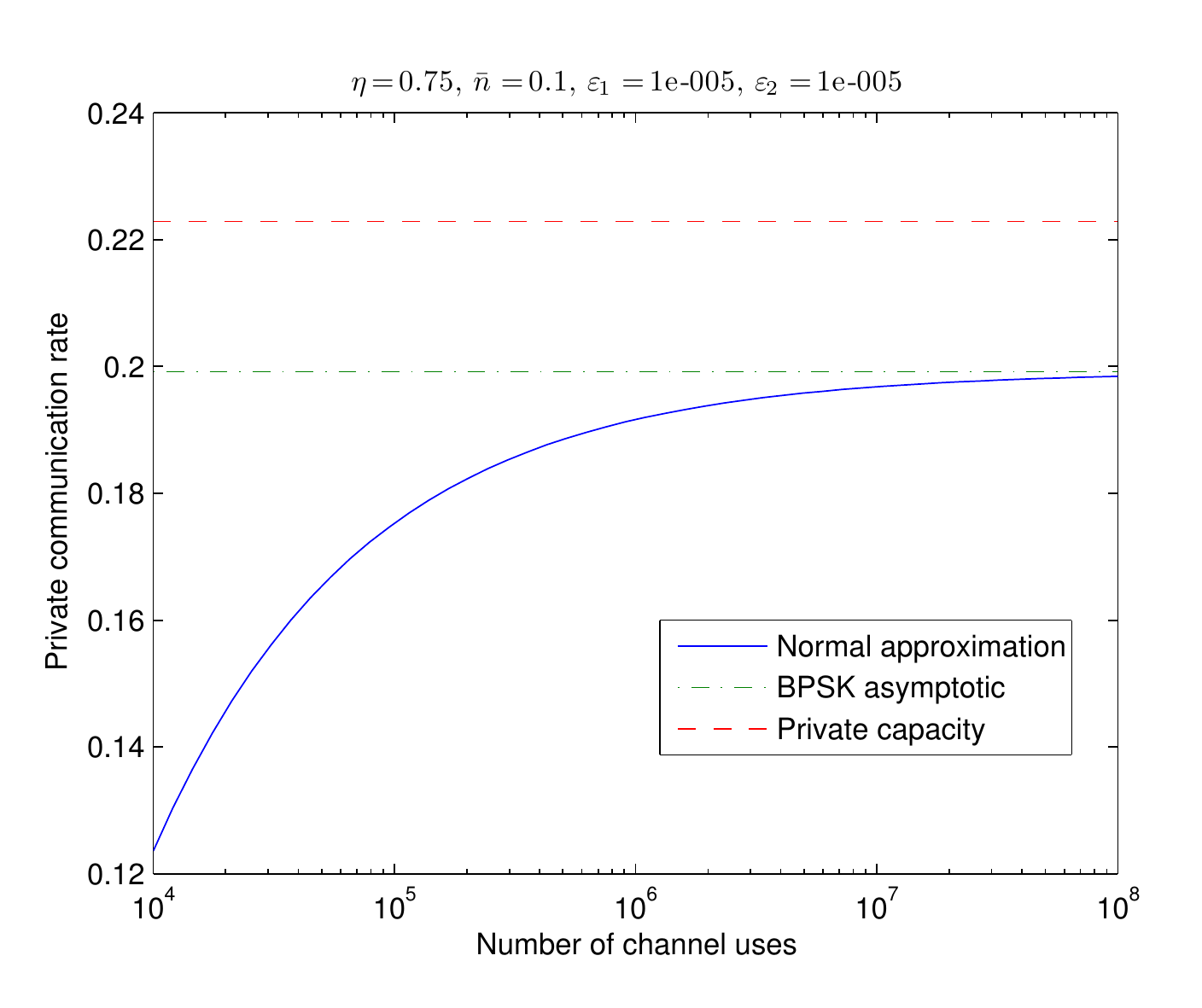}}
\qquad
\subfloat[]{\includegraphics[width=.47\columnwidth]{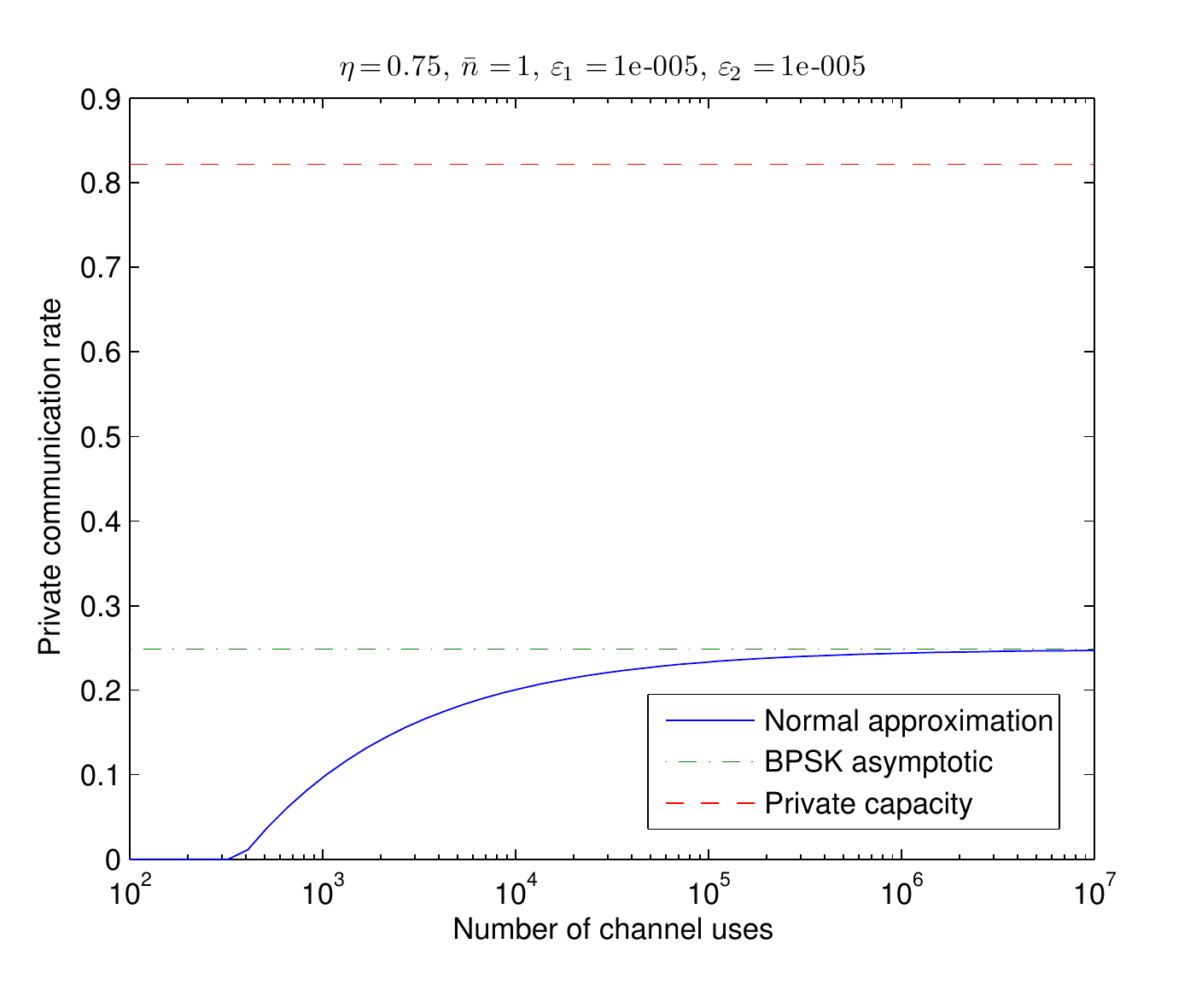}}\newline%
\subfloat[]{\includegraphics[width=.47\columnwidth]{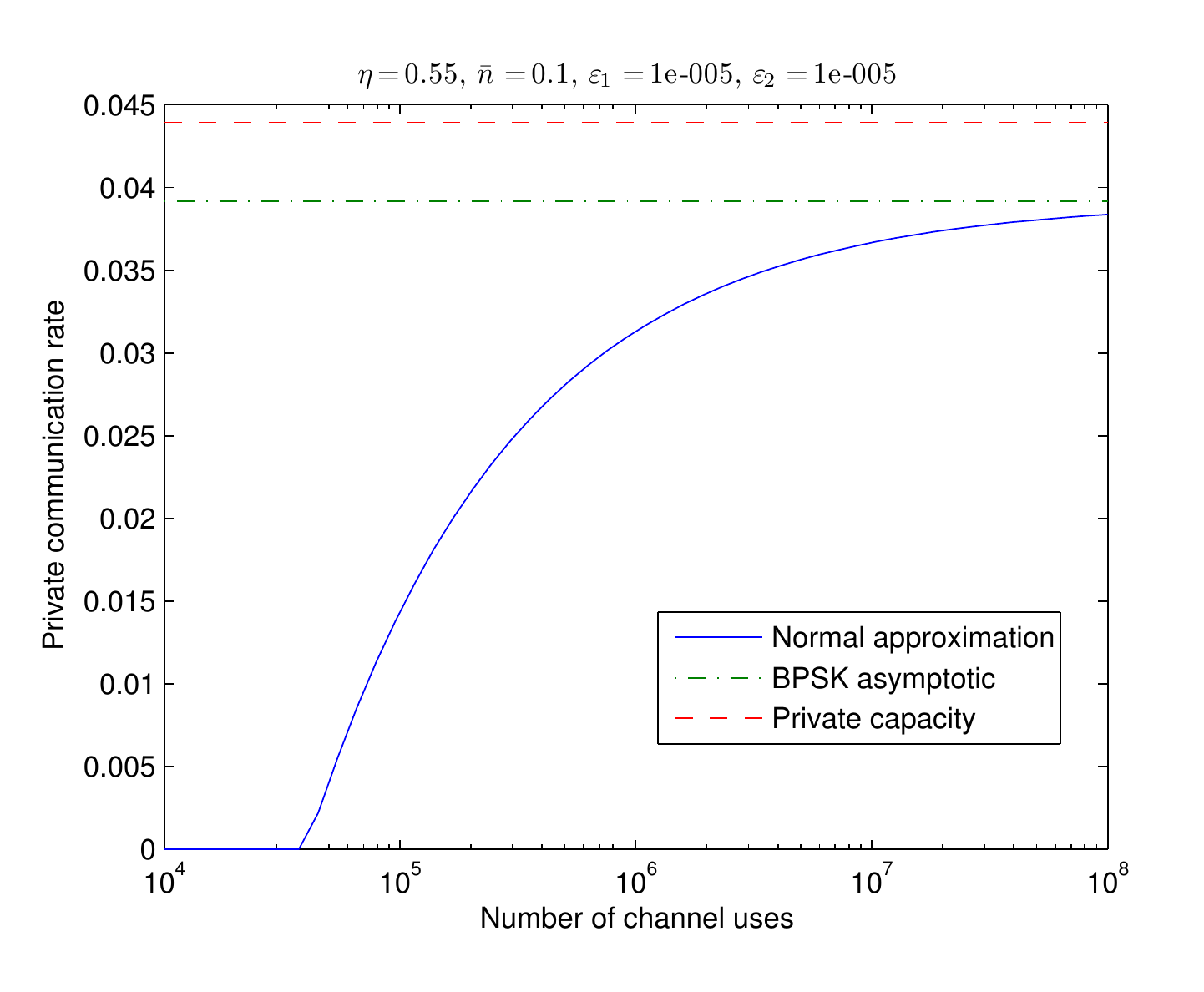}}
\qquad
\subfloat[]{\includegraphics[width=.47\columnwidth]{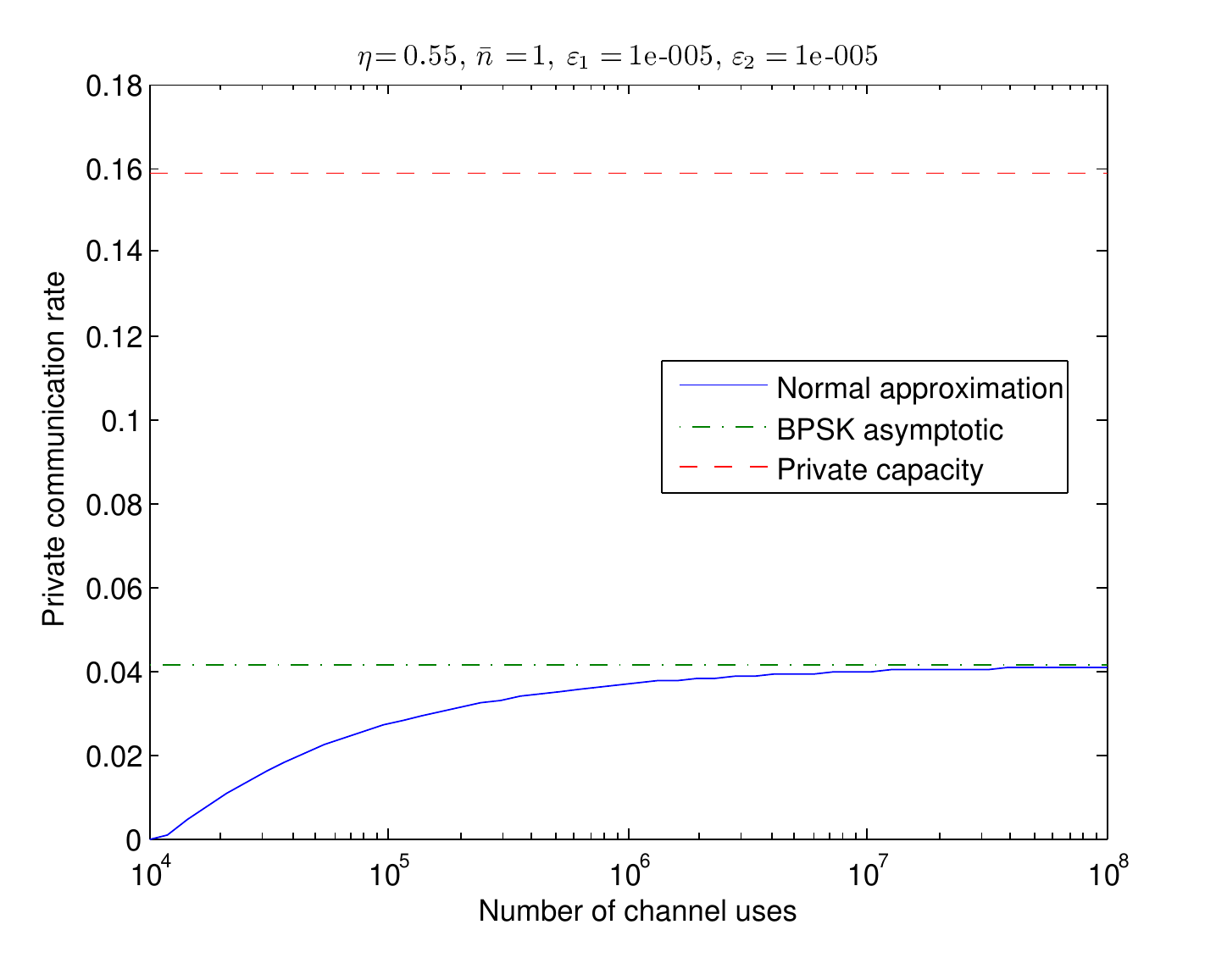}}
\end{center}
\caption{The figures plot the normal approximation for non-asymptotic BPSK private communication (using \eqref{eq:pure-loss-2nd-order}), the asymptotic limit for BPSK, and the asymptotic energy-constrained private capacity for various values of the channel transmissivity $\eta$,
the mean photon number $\bar{n}$, $\varepsilon_1$,
and $\varepsilon_2$.}%
\label{fig:results}%
\end{figure}

\section{Conclusion\label{sec:concl}}

This paper establishes a lower bound on the $\varepsilon$-one-shot private
classical capacity of a cq wiretap channel, which in turn leads to a lower
bound on the second-order coding rate for private communication over an
i.i.d.~cq wiretap channel. The main techniques used are position-based
decoding \cite{AJW17}\ in order to guarantee that Bob can decode reliably and
convex splitting \cite{ADJ17}\ to guarantee that Eve cannot determine which
message Alice transmitted. It is my opinion that these two methods represent a
powerful approach to quantum information theory, having already been used
effectively in a variety of contexts in \cite{ADJ17,AJW17}.

For future work, it would be good to improve upon the lower bounds given here.
Extensions of the methods of \cite{YSP16} and \cite{TB16}\ might be helpful in
this endeavor.\bigskip

\textit{Note}: After the completion of the results in the present paper,
Naqueeb Warsi informed the author of an unpublished result from \cite{Warsi17}%
,  which establishes a lower bound on the $\varepsilon$-one-shot private
capacity of a cq wiretap channel in terms of a difference of the hypothesis
testing mutual information and a smooth max-mutual information.

\bigskip

\textbf{Acknowledgements.} I am grateful to Anurag Anshu, Saikat Guha, Rahul
Jain, Haoyu Qi, Qingle Wang, and Naqueeb Warsi for discussions related to the
topic of this paper. I acknowledge support from the Office of Naval Research
and the National Science Foundation.

\appendix

\section{Proof of convex-split lemma}

\label{app:convex-split}

For the sake of completeness, this appendix features a proof of
Lemma~\ref{thm:convex-split}. Let $\widetilde{\rho}_{AB}$ be the optimizer for%
\begin{equation}
\lambda^{\ast}\equiv\widetilde{I}_{\max}^{\sqrt{\varepsilon}-\eta}(B;A)_{\rho
}=\inf_{\rho_{AB}^{\prime}\ :\ P(\rho_{AB}^{\prime},\rho_{AB})\leq
\sqrt{\varepsilon}-\eta}D_{\max}(\rho_{AB}^{\prime}\Vert\rho_{A}\otimes
\rho_{B}^{\prime}). \label{eq:def-lambda-star}%
\end{equation}
We take $\widetilde{\rho}_{B}$ as the marginal of $\widetilde{\rho}_{AB}$. We
define the following state, which we think of as an approximation to
$\tau_{A_{1}\cdots A_{K}B}$:%
\begin{equation}
\widetilde{\tau}_{A_{1}\cdots A_{K}B}\equiv\frac{1}{K}\sum_{k=1}^{K}%
\rho_{A_{1}}\otimes\cdots\otimes\rho_{A_{k-1}}\otimes\widetilde{\rho}_{A_{k}%
B}\otimes\rho_{A_{k+1}}\otimes\cdots\otimes\rho_{A_{K}}.
\end{equation}
In fact, it is a good approximation if $\sqrt{\varepsilon}-\eta$ is
small:\ Consider from joint concavity of the root fidelity that%
\begin{align}
&  \sqrt{F}(\widetilde{\tau}_{A_{1}\cdots A_{K}B},\tau_{A_{1}\cdots A_{K}%
B})\nonumber\\
&  \geq\frac{1}{K}\sum_{k=1}^{K}\sqrt{F}(\rho_{A_{1}}\otimes\cdots\otimes
\rho_{A_{k-1}}\otimes\widetilde{\rho}_{A_{k}B}\otimes\rho_{A_{k+1}}%
\otimes\cdots\otimes\rho_{A_{K}},\nonumber\\
&  \qquad\qquad\qquad\qquad\rho_{A_{1}}\otimes\cdots\otimes\rho_{A_{k-1}%
}\otimes\rho_{A_{k}B}\otimes\rho_{A_{k+1}}\otimes\cdots\otimes\rho_{A_{K}})\\
&  =\frac{1}{K}\sum_{k=1}^{K}\sqrt{F}(\widetilde{\rho}_{A_{k}B},\rho_{A_{k}%
B})\\
&  =\sqrt{F}(\widetilde{\rho}_{AB},\rho_{AB}),
\end{align}
which in turn implies that%
\begin{equation}
F(\widetilde{\tau}_{A_{1}\cdots A_{K}B},\tau_{A_{1}\cdots A_{K}B})\geq
F(\widetilde{\rho}_{AB},\rho_{AB}). \label{eq:scrambled-state-fid}%
\end{equation}
So the inequality in \eqref{eq:scrambled-state-fid}, the definition of the
purified distance, and the fact that $P(\rho_{AB},\widetilde{\rho}_{AB}%
)\leq\sqrt{\varepsilon}-\eta$ imply that%
\begin{equation}
P(\widetilde{\tau}_{A_{1}\cdots A_{K}B},\tau_{A_{1}\cdots A_{K}B})\leq
\sqrt{\varepsilon}-\eta.
\end{equation}
Let $\omega=\sum_{y}p_{Y}(y)\omega^{y}$, for $p_{Y}$ a probability
distribution and $\{\omega^{y}\}_{y}$ a set of states. Then the following
property holds for quantum relative entropy and a state $\kappa$ such that
$\operatorname{supp}(\omega)\subseteq\operatorname{supp}(\kappa)$:%
\begin{equation}
D(\omega\Vert\kappa)=\sum_{y}p_{Y}(y)\left[  D(\omega^{y}\Vert\kappa
)-D(\omega^{y}\Vert\omega)\right]  . \label{eq:rel-ent-prop}%
\end{equation}
Applying \eqref{eq:rel-ent-prop}, it follows that%
\begin{multline}
D(\widetilde{\tau}_{A_{1}\cdots A_{K}B}\Vert\rho_{A_{1}}\otimes\cdots
\otimes\rho_{A_{K}}\otimes\widetilde{\rho}_{B})\label{eq:using-rel-ent-prop}\\
=\frac{1}{K}\sum_{k=1}^{K}D(\rho_{A_{1}}\otimes\cdots\otimes\rho_{A_{k-1}%
}\otimes\rho_{A_{k}B}\otimes\rho_{A_{k+1}}\otimes\cdots\otimes\rho_{A_{K}%
}\Vert\rho_{A_{1}}\otimes\cdots\otimes\rho_{A_{K}}\otimes\widetilde{\rho}%
_{B})\\
-\frac{1}{K}\sum_{k=1}^{K}D(\rho_{A_{1}}\otimes\cdots\otimes\rho_{A_{k-1}%
}\otimes\rho_{A_{k}B}\otimes\rho_{A_{k+1}}\otimes\cdots\otimes\rho_{A_{K}%
}\Vert\widetilde{\tau}_{A_{1}\cdots A_{K}B}).
\end{multline}
The first term in \eqref{eq:using-rel-ent-prop} on the right-hand side of the
equality simplifies as%
\begin{align}
&  \frac{1}{K}\sum_{k=1}^{K}D(\rho_{A_{1}}\otimes\cdots\otimes\rho_{A_{k-1}%
}\otimes\rho_{A_{k}B}\otimes\rho_{A_{k+1}}\otimes\cdots\otimes\rho_{A_{K}%
}\Vert\rho_{A_{1}}\otimes\cdots\otimes\rho_{A_{K}}\otimes\widetilde{\rho}%
_{B})\nonumber\\
&  =\frac{1}{K}\sum_{k=1}^{K}D(\rho_{A_{k}B}\Vert\rho_{A_{k}}\otimes
\widetilde{\rho}_{B})\\
&  =D(\rho_{AB}\Vert\rho_{A}\otimes\widetilde{\rho}_{B}).
\label{eq:convex-split-dev-key-1}%
\end{align}
We now lower bound the last term in \eqref{eq:using-rel-ent-prop}. Consider
that a partial trace over systems $A_{1}$, \ldots, $A_{k-1}$, $A_{k+1}$,
\ldots, $A_{K}$ gives%
\begin{multline}
D(\rho_{A_{1}}\otimes\cdots\otimes\rho_{A_{k-1}}\otimes\rho_{A_{k}B}%
\otimes\rho_{A_{k+1}}\otimes\cdots\otimes\rho_{A_{K}}\Vert\widetilde{\tau
}_{A_{1}\cdots A_{K}B})\label{eq:lower-bnd-convex-split-middle-step}\\
\geq D(\rho_{A_{k}B}\Vert\widetilde{\tau}_{A_{k}B})=D(\rho_{AB}\Vert\left[
1/K\right]  \widetilde{\rho}_{AB}+\left[  1-1/K\right]  \rho_{A}%
\otimes\widetilde{\rho}_{B}),
\end{multline}
where the equality follows because $\widetilde{\tau}_{A_{k}B}=\left[
1/K\right]  \widetilde{\rho}_{AB}+\left[  1-1/K\right]  \rho_{A}%
\otimes\widetilde{\rho}_{B}$. Thus, averaging the inequality in
\eqref{eq:lower-bnd-convex-split-middle-step}\ over $k$ implies that%
\begin{multline}
\frac{1}{K}\sum_{k=1}^{K}D(\rho_{A_{1}}\otimes\cdots\otimes\rho_{A_{k-1}%
}\otimes\rho_{A_{k}B}\otimes\rho_{A_{k+1}}\otimes\cdots\otimes\rho_{A_{K}%
}\Vert\widetilde{\tau}_{A_{1}\cdots A_{K}B}) \label{eq:convex-split-dev-key-2}%
\\
\geq D(\rho_{AB}\Vert\left[  1/K\right]  \widetilde{\rho}_{AB}+\left[
1-1/K\right]  \rho_{A}\otimes\widetilde{\rho}_{B}),
\end{multline}
Putting together \eqref{eq:using-rel-ent-prop},
\eqref{eq:convex-split-dev-key-1}, and \eqref{eq:convex-split-dev-key-2}, we
find that%
\begin{multline}
D(\widetilde{\tau}_{A_{1}\cdots A_{K}B}\Vert\rho_{A_{1}}\otimes\cdots
\otimes\rho_{A_{K}}\otimes\rho_{B})\label{eq:CS-bnds-1}\\
\leq D(\rho_{AB}\Vert\rho_{A}\otimes\widetilde{\rho}_{B})-D(\rho_{AB}%
\Vert\left[  1/K\right]  \widetilde{\rho}_{AB}+\left[  1-1/K\right]  \rho
_{A}\otimes\widetilde{\rho}_{B})
\end{multline}
By the definition of $\lambda^{\ast}$ in \eqref{eq:def-lambda-star}, we have
that%
\begin{equation}
\widetilde{\rho}_{AB}\leq2^{\lambda^{\ast}}\rho_{A}\otimes\widetilde{\rho}%
_{B}, \label{eq:op-ineq-d-max}%
\end{equation}
which means that%
\begin{equation}
\left[  1/K\right]  \widetilde{\rho}_{AB}+\left[  1-1/K\right]  \rho
_{A}\otimes\widetilde{\rho}_{B}\leq\left[  1+(2^{\lambda^{\ast}}-1)/K\right]
\rho_{A}\otimes\widetilde{\rho}_{B}. \label{eq:op-ineq-d-max-CS}%
\end{equation}
An important property of quantum relative entropy is that $D(\omega\Vert
\tau)\geq D(\omega\Vert\tau^{\prime})$ if $\tau\leq\tau^{\prime}$. Applying it
to \eqref{eq:op-ineq-d-max-CS}\ and the right-hand side of
\eqref{eq:CS-bnds-1}, we get that%
\begin{align}
&  D(\rho_{AB}\Vert\rho_{A}\otimes\widetilde{\rho}_{B})-D(\rho_{AB}%
\Vert\left[  1/K\right]  \widetilde{\rho}_{AB}+\left[  1-1/K\right]  \rho
_{A}\otimes\widetilde{\rho}_{B})\nonumber\\
&  \leq D(\rho_{AB}\Vert\rho_{A}\otimes\widetilde{\rho}_{B})-D(\rho_{AB}%
\Vert\left[  1+(2^{\lambda^{\ast}}-1)/K\right]  \rho_{A}\otimes\widetilde
{\rho}_{B})\label{eq:CS-almost-1}\\
&  =D(\rho_{AB}\Vert\rho_{A}\otimes\widetilde{\rho}_{B})-D(\rho_{AB}\Vert
\rho_{A}\otimes\widetilde{\rho}_{B})+\log_{2}\left[  1+(2^{\lambda^{\ast}%
}-1)/K\right] \\
&  =\log_{2}\left[  1+(2^{\lambda^{\ast}}-1)/K\right]  .
\label{eq:CS-almost-3}%
\end{align}
Then the well known inequality $D(\omega\Vert\tau)\geq-\log_{2}F(\omega,\tau
)$, \eqref{eq:CS-bnds-1}, and \eqref{eq:CS-almost-1}--\eqref{eq:CS-almost-3}
imply that%
\begin{equation}
-\log_{2}F(\widetilde{\tau}_{A_{1}\cdots A_{K}B},\rho_{A_{1}}\otimes
\cdots\otimes\rho_{A_{K}}\otimes\widetilde{\rho}_{B})\leq\log_{2}\left[
1+(2^{\lambda^{\ast}}-1)/K\right]  ,
\end{equation}
which in turn implies that%
\begin{equation}
F(\widetilde{\tau}_{A_{1}\cdots A_{K}B},\rho_{A_{1}}\otimes\cdots\otimes
\rho_{A_{K}}\otimes\widetilde{\rho}_{B})\geq\frac{1}{1+\frac{2^{\lambda^{\ast
}}-1}{K}}=1-\frac{2^{\lambda^{\ast}}-1}{K+2^{\lambda^{\ast}}-1}\geq
1-\frac{2^{\lambda^{\ast}}}{K}.
\end{equation}
So if we pick $K$ such that%
\begin{equation}
\log_{2}K=\inf_{\rho_{AB}^{\prime}\ :\ P(\rho_{AB}^{\prime},\rho_{AB}%
)\leq\sqrt{\varepsilon}-\eta}D_{\max}(\rho_{AB}^{\prime}\Vert\rho_{A}%
\otimes\rho_{B}^{\prime})+2\log_{2}\left(  \frac{1}{\eta}\right)  ,
\end{equation}
then we are guaranteed that%
\begin{equation}
P(\widetilde{\tau}_{A_{1}\cdots A_{K}B},\rho_{A_{1}}\otimes\cdots\otimes
\rho_{A_{K}}\otimes\widetilde{\rho}_{B})\leq\eta.
\end{equation}
By the triangle inequality for the purified distance, we then get that%
\begin{align}
&  P(\tau_{A_{1}\cdots A_{K}B},\rho_{A_{1}}\otimes\cdots\otimes\rho_{A_{K}%
}\otimes\widetilde{\rho}_{B})\nonumber\\
&  \leq P(\tau_{A_{1}\cdots A_{K}B},\widetilde{\tau}_{A_{1}\cdots A_{K}%
B})+P(\widetilde{\tau}_{A_{1}\cdots A_{K}B},\rho_{A_{1}}\otimes\cdots
\otimes\rho_{A_{K}}\otimes\widetilde{\rho}_{B})\\
&  \leq\left(  \sqrt{\varepsilon}-\eta\right)  +\eta=\sqrt{\varepsilon}.
\end{align}
This concludes the proof.

\bibliographystyle{alpha}
\bibliography{Ref}

\newcommand{\etalchar}[1]{$^{#1}$}
\begin{thebibliography}{HHHO09}

\bibitem[ADJ17]{ADJ17}
Anurag Anshu, Vamsi~Krishna Devabathini, and Rahul Jain.
\newblock Quantum message compression with applications.
\newblock February 2017.
\newblock arXiv:1410.3031v4.

\bibitem[AJW17]{AJW17}
Anurag Anshu, Rahul Jain, and Naqueeb~Ahmad Warsi.
\newblock One shot entanglement assisted classical and quantum communication
  over noisy quantum channels: A hypothesis testing and convex split approach.
\newblock February 2017.
\newblock arXiv:1702.01940.

\bibitem[BB84]{bb84}
Charles~H. Bennett and Gilles Brassard.
\newblock Quantum cryptography: Public key distribution and coin tossing.
\newblock In {\em Proceedings of IEEE International Conference on Computers
  Systems and Signal Processing}, pages 175--179, Bangalore, India, December
  1984.

\bibitem[BD10]{BD10}
Francesco Buscemi and Nilanjana Datta.
\newblock The quantum capacity of channels with arbitrarily correlated noise.
\newblock {\em IEEE Transactions on Information Theory}, 56(3):1447--1460,
  March 2010.
\newblock arXiv:0902.0158.

\bibitem[BDL16]{BDL15}
Salman Beigi, Nilanjana Datta, and Felix Leditzky.
\newblock Decoding quantum information via the {Petz} recovery map.
\newblock {\em Journal of Mathematical Physics}, 57(8):082203, August 2016.
\newblock arXiv:1504.04449.

\bibitem[CK78]{CK78}
Imre Csisz\'{a}r and Janos K\"orner.
\newblock Broadcast channels with confidential messages.
\newblock {\em IEEE Transactions on Information Theory}, 24(3):339--348, May
  1978.

\bibitem[CWY04]{1050633}
Ning Cai, Andreas Winter, and Raymond~W. Yeung.
\newblock Quantum privacy and quantum wiretap channels.
\newblock {\em Problems of Information Transmission}, 40(4):318--336, October
  2004.

\bibitem[Dat09]{D09}
Nilanjana Datta.
\newblock Min- and max-relative entropies and a new entanglement monotone.
\newblock {\em IEEE Transactions on Information Theory}, 55(6):2816--2826, June
  2009.
\newblock arXiv:0803.2770.

\bibitem[Dev05]{ieee2005dev}
Igor Devetak.
\newblock The private classical capacity and quantum capacity of a quantum
  channel.
\newblock {\em IEEE Transactions on Information Theory}, 51(1):44--55, January
  2005.
\newblock arXiv:quant-ph/0304127.

\bibitem[DHO16]{DHO16}
Nilanjana Datta, Min-Hsiu Hsieh, and Jonathan Oppenheim.
\newblock An upper bound on the second order asymptotic expansion for the
  quantum communication cost of state redistribution.
\newblock {\em Journal of Mathematical Physics}, 57(5):052203, May 2016.
\newblock arXiv:1409.4352.

\bibitem[DL15]{DL15}
Nilanjana Datta and Felix Leditzky.
\newblock Second-order asymptotics for source coding, dense coding, and
  pure-state entanglement conversions.
\newblock {\em IEEE Transactions on Information Theory}, 61(1):582--608,
  January 2015.
\newblock arXiv:1403.2543.

\bibitem[DTW16]{DTW14}
Nilanjana Datta, Marco Tomamichel, and Mark~M. Wilde.
\newblock On the second-order asymptotics for entanglement-assisted
  communication.
\newblock {\em Quantum Information Processing}, 15(6):2569--2591, June 2016.
\newblock arXiv:1405.1797.

\bibitem[GGL{\etalchar{+}}04]{GGLMSY04}
Vittorio Giovannetti, Saikat Guha, Seth Lloyd, Lorenzo Maccone, Jeffrey~H.
  Shapiro, and Horace~P. Yuen.
\newblock Classical capacity of the lossy bosonic channel: The exact solution.
\newblock {\em Physical Review Letters}, 92(2):027902, January 2004.
\newblock arXiv:quant-ph/0308012.

\bibitem[GLN05]{GLN04}
Alexei Gilchrist, Nathan~K. Langford, and Michael~A. Nielsen.
\newblock Distance measures to compare real and ideal quantum processes.
\newblock {\em Physical Review A}, 71(6):062310, June 2005.
\newblock arXiv:quant-ph/0408063.

\bibitem[GW12]{GW12}
Saikat Guha and Mark~M. Wilde.
\newblock Polar coding to achieve the {Holevo} capacity of a pure-loss optical
  channel.
\newblock {\em Proceedings of the 2012 IEEE International Symposium on
  Information Theory}, pages 546--550, 2012.
\newblock arXiv:1202.0533.

\bibitem[Hay06]{H06}
Masahito Hayashi.
\newblock General nonasymptotic and asymptotic formulas in channel
  resolvability and identification capacity and their application to the
  wiretap channel.
\newblock {\em IEEE Transactions on Information Theory}, 52(4):1562--1575,
  April 2006.
\newblock arXiv:cs/0503088.

\bibitem[Hay13]{H13}
Masahito Hayashi.
\newblock Tight exponential analysis of universally composable privacy
  amplification and its applications.
\newblock {\em IEEE Transactions on Information Theory}, 59(11):7728--7746,
  November 2013.
\newblock arXiv:1010.1358.

\bibitem[HHHO09]{HHHO09}
Karol Horodecki, Michal Horodecki, Pawel Horodecki, and Jonathan Oppenheim.
\newblock General paradigm for distilling classical key from quantum states.
\newblock {\em IEEE Transactions on Information Theory}, 55(4):1898--1929,
  April 2009.
\newblock arXiv:quant-ph/0506189.

\bibitem[HN03]{HN03}
Masahito Hayashi and Hiroshi Nagaoka.
\newblock General formulas for capacity of classical-quantum channels.
\newblock {\em IEEE Transactions on Information Theory}, 49(7):1753--1768, July
  2003.
\newblock arXiv:quant-ph/0206186.

\bibitem[Li14]{li12}
Ke~Li.
\newblock Second order asymptotics for quantum hypothesis testing.
\newblock {\em Annals of Statistics}, 42(1):171--189, February 2014.
\newblock arXiv:1208.1400.

\bibitem[PPV10]{polyanskiy10}
Yury Polyanskiy, H.~Vincent Poor, and Sergio Verd\'{u}.
\newblock Channel coding rate in the finite blocklength regime.
\newblock {\em IEEE Transactions on Information Theory}, 56(5):2307--2359, May
  2010.

\bibitem[Ras02]{R02}
Alexey~E. Rastegin.
\newblock Relative error of state-dependent cloning.
\newblock {\em Physical Review A}, 66(4):042304, October 2002.

\bibitem[Ras03]{R03}
Alexey~E. Rastegin.
\newblock A lower bound on the relative error of mixed-state cloning and
  related operations.
\newblock {\em Journal of Optics B: Quantum and Semiclassical Optics},
  5(6):S647, December 2003.
\newblock arXiv:quant-ph/0208159.

\bibitem[Ras06]{R06}
Alexey~E. Rastegin.
\newblock Sine distance for quantum states.
\newblock February 2006.
\newblock arXiv:quant-ph/0602112.

\bibitem[RR11]{RR11}
Joseph~M. Renes and Renato Renner.
\newblock Noisy channel coding via privacy amplification and information
  reconciliation.
\newblock {\em IEEE Transactions on Information Theory}, 57(11):7377--7385, Nov
  2011.

\bibitem[Ser17]{S17}
Alessio Serafini.
\newblock {\em Quantum Continuous Variables}.
\newblock CRC Press, 2017.

\bibitem[Tan12]{T12}
Vincent Y.~F. Tan.
\newblock Achievable second-order coding rates for the wiretap channel.
\newblock In {\em 2012 IEEE International Conference on Communication Systems
  (ICCS)}, pages 65--69, November 2012.

\bibitem[TB16]{TB16}
Mehrdad Tahmasbi and Matthieu~R. Bloch.
\newblock Second order asymptotics for degraded wiretap channels: How good are
  existing codes?
\newblock In {\em 2016 54th Annual Allerton Conference on Communication,
  Control, and Computing (Allerton)}, pages 830--837, September 2016.

\bibitem[TBR16]{TBR15}
Marco Tomamichel, Mario Berta, and Joseph~M. Renes.
\newblock Quantum coding with finite resources.
\newblock {\em Nature Communications}, 7:11419, May 2016.
\newblock arXiv:1504.04617.

\bibitem[TCR09]{TCR09}
Marco Tomamichel, Roger Colbeck, and Renato Renner.
\newblock A fully quantum asymptotic equipartition property.
\newblock {\em IEEE Transactions on Information Theory}, 55(12):5840--5847,
  December 2009.
\newblock arXiv:0811.1221.

\bibitem[TH13]{TH12}
Marco Tomamichel and Masahito Hayashi.
\newblock A hierarchy of information quantities for finite block length
  analysis of quantum tasks.
\newblock {\em IEEE Transactions on Information Theory}, 59(11):7693--7710,
  November 2013.
\newblock arXiv:1208.1478.

\bibitem[TT15]{TT13}
Marco Tomamichel and Vincent Y.~F. Tan.
\newblock Second-order asymptotics for the classical capacity of image-additive
  quantum channels.
\newblock {\em Communications in Mathematical Physics}, 338(1):103--137, August
  2015.
\newblock arXiv:1308.6503.

\bibitem[Uhl76]{U76}
Armin Uhlmann.
\newblock The ``transition probability'' in the state space of a *-algebra.
\newblock {\em Reports on Mathematical Physics}, 9(2):273--279, 1976.

\bibitem[Ume62]{U62}
Hisaharu Umegaki.
\newblock Conditional expectations in an operator algebra {IV} (entropy and
  information).
\newblock {\em Kodai Mathematical Seminar Reports}, 14(2):59--85, 1962.

\bibitem[War15]{Warsi17}
Naqueeb~Ahmad Warsi.
\newblock {\em One-shot bounds in classical and quantum information theory}.
\newblock PhD thesis, Tata Institute of Fundamental Research, Mumbai, India,
  December 2015.
\newblock Not publicly available; communicated by email on March 5, 2017.

\bibitem[Wil16]{W15book}
Mark~M. Wilde.
\newblock {\em From Classical to Quantum Shannon Theory}.
\newblock March 2016.
\newblock arXiv:1106.1445v7.

\bibitem[WQ16]{wilde2016energy}
Mark~M. Wilde and Haoyu Qi.
\newblock Energy-constrained private and quantum capacities of quantum
  channels.
\newblock September 2016.
\newblock arXiv:1609.01997.

\bibitem[WR12]{WR12}
Ligong Wang and Renato Renner.
\newblock One-shot classical-quantum capacity and hypothesis testing.
\newblock {\em Physical Review Letters}, 108(20):200501, May 2012.
\newblock arXiv:1007.5456.

\bibitem[WTB17]{WTB16}
Mark~M. Wilde, Marco Tomamichel, and Mario Berta.
\newblock Converse bounds for private communication over quantum channels.
\newblock {\em IEEE Transactions on Information Theory}, 63(3):1792--1817,
  March 2017.
\newblock arXiv:1602.08898.

\bibitem[Wyn75]{W75}
Aaron~D. Wyner.
\newblock The wire-tap channel.
\newblock {\em Bell System Technical Journal}, 54(8):1355--1387, October 1975.

\bibitem[YAG13]{YAG13}
Mohammad~Hossein Yassaee, Mohammad~Reza Aref, and Amin Gohari.
\newblock Non-asymptotic output statistics of random binning and its
  applications.
\newblock In {\em 2013 IEEE International Symposium on Information Theory},
  pages 1849--1853, July 2013.
\newblock arXiv:1303.0695.

\bibitem[YSP16]{YSP16}
Wei Yang, Rafael~F. Schaefer, and H.~Vincent Poor.
\newblock Finite-blocklength bounds for wiretap channels.
\newblock In {\em 2016 IEEE International Symposium on Information Theory
  (ISIT)}, pages 3087--3091, July 2016.
\newblock arXiv:1601.06055.

\end{thebibliography}

\end{document}